\documentclass[a4paper,USenglish,authorcolumns,pdfa]{lipics-v2021}
\usepackage{hyperref}
\usepackage{color}
\usepackage{algorithm}
\usepackage[noend]{algpseudocode}
\usepackage{xspace}
\usepackage[textsize=small,disable]{todonotes}
\usepackage{booktabs}
\usepackage{enumerate}
\usepackage{subcaption}

\hideLIPIcs

\definecolor{darkgreen}{rgb}{0,0.5,0}
\definecolor{darkblue}{rgb}{0,0,0.5}
\hypersetup{
	unicode=false,
	colorlinks=true,
	linkcolor=darkblue,
	citecolor=darkgreen,
	filecolor=magenta,
	urlcolor=cyan
}
\newcommand{\LOCAL}{\ensuremath{\mathsf{LOCAL}}\xspace}
\newcommand{\CONGEST}{\ensuremath{\mathsf{CONGEST}}\xspace}
\newcommand{\BCONGEST}{\ensuremath{\mathsf{Broadcast\mbox{-}CONGEST}}\xspace}
\newcommand{\NCC}{\ensuremath{\mathsf{NCC}}\xspace}
\newcommand{\HYBRID}{\ensuremath{\mathsf{HYBRID}}\xspace}

\newcommand{\UDG}{\ensuremath{\mathsf{UDG}}\xspace}

\newcommand{\bigO}{\smash{\ensuremath{O}}}
\newcommand{\tilO}{\smash{\ensuremath{\widetilde{O}}}}

\newcommand{\eps}{\varepsilon}

\newcommand{\m}{\!-\!}

\newcommand{\interior}[1]{%
	{\kern0pt#1}^{\mathrm{o}}%
}

\DeclareMathOperator{\polylog}{polylog}

\DeclareMathOperator*{\argmin}{arg\,min}
\DeclareMathOperator{\dist}{\,dist}
\DeclareMathOperator{\hop}{\,hop}
\DeclareMathOperator{\incomp}{\nsim}

\title{Near-Shortest Path Routing in Hybrid Communication Networks
}
\titlerunning{Near-Shortest Path Routing in Hybrid Communication Networks}

\author{Sam Coy}{University of Warwick, United Kingdom}{S.Coy@warwick.ac.uk}{}{}

\author{Artur Czumaj}{University of Warwick, United Kingdom}{A.Czumaj@warwick.ac.uk}{}{}

\author{Michael Feldmann}{Paderborn University, Germany}{michael.feldmann@upb.de}{}{}

\author{Kristian Hinnenthal}{Paderborn University, Germany}{krijan@mail.upb.de}{}{}

\author{Fabian Kuhn}{University of Freiburg, Germany}{kuhn@cs.uni-freiburg.de}{}{}

\author{Christian Scheideler}{Paderborn University, Germany}{scheideler@upb.de}{}{}

\author{Philipp Schneider}{University of Freiburg, Germany}{philipp.schneider@cs.uni-freiburg.de}{}{}

\author{Martijn Struijs}{TU Eindhoven, The Netherlands}{m.a.c.struijs@tue.nl}{}{}

\authorrunning{S. Coy, A. Czumaj, M. Feldmann, K. Hinnenthal, F. Kuhn, C. Scheideler, P. Schneider, M. Struijs}

\Copyright{Sam Coy, Artur Czumaj, Michael Feldmann, Kristian Hinnenthal, Fabian Kuhn, Christian Scheideler, Philipp Schnieder, and Martijn Struijs}

\ccsdesc[500]{Theory of computation~Distributed algorithms} %TODO mandatory: Please choose ACM 2012 classifications from https://dl.acm.org/ccs/ccs_flat.cfm

\keywords{Hybrid networks, overlay networks}

\supplement{}%optional, e.g., related research data, source code, ... hosted on a repository like zenodo, figshare, GitHub, ...

\funding{	This work has been partially supported by the German Research Foundation (DFG) within the Collaborative Research Center 901 ''On-The-Fly Computing'' under the project number 160364472-SFB901, by the Centre for Discrete Mathematics and its Applications (DIMAP), by EPSRC studentship, and by EPSRC award EP/V01305X/1.}%optional, to capture a funding statement, which applies to all authors. Please enter author specific funding statements as fifth argument of the \author macro.

\nolinenumbers %uncomment to disable line numbering

\EventEditors{Quentin Bramas, Vincent Gramoli, and Alessia Milani}
\EventNoEds{3}
\EventLongTitle{25th International Conference on Principles of Distributed Systems (OPODIS 2021)}
\EventShortTitle{OPODIS 2021}
\EventAcronym{OPODIS}
\EventYear{2021}
\EventDate{December 13--15, 2021}
\EventLocation{Strasbourg, France}
\EventLogo{}
\SeriesVolume{217}
\ArticleNo{25}

\begin{document}

\maketitle

\begin{abstract}
    Hybrid networks, i.e., networks that leverage different means of communication, become ever more widespread. To allow theoretical study of such networks, [Augustine et al., SODA'20] introduced the \HYBRID model, which is based on the concept of synchronous message passing and uses two fundamentally different principles of communication: a \emph{local} mode, which allows every node to exchange one message per round with each neighbor in a \emph{local communication graph}; and a \emph{global} mode where \emph{any pair} of nodes can exchange messages, but only \emph{few such exchanges} can take place per round.
    A sizable portion of the previous research for the \HYBRID model revolves around basic communication primitives and computing distances or shortest paths in networks. In this paper, we extend this study to a related fundamental problem of \emph{computing compact routing schemes for near-shortest paths} in the local communication graph. We demonstrate that, for the case where the local communication graph is a \emph{unit-disc graph} with $n$ nodes that is realized in the plane and has \emph{no radio holes}, we can deterministically compute a routing scheme that has constant stretch and uses labels and local routing tables of size $O(\log n)$ bits in only $O(\log n)$ rounds.
\end{abstract}

\section{Introduction}

Humans naturally communicate in a hybrid fashion by making use of broadcast services, emails, phones, or simply face-to-face communication. Thus, it seems natural to study hybrid communication also in distributed systems. But fundamental research in this area is still in its infancy, even though
there are several examples where hybrid communication is already exploited in practice. For instance, in modern data centers, wired communication networks are combined with high-speed wireless communication to reduce wire length or increase bandwidth without adding congestion to the wired network \cite{FoersterS19}.
This paper focuses on hybrid \emph{wireless} networks: networks that combine ad-hoc, WLAN-based connections (the \emph{local} network) with connections via a cellular or satellite infrastructure (the \emph{global} network). These can be realized, for instance, by smartphones, since they support both communication modes
and solutions for smartphone ad hoc networks have been around for almost a decade. Connections in the local network can transfer large amounts of data cheaply, but have limited range, while global connections can transmit data between any pair of devices, but typically with bandwidth restrictions and additional costs. So ideally, global communication should be reserved for exchanging control messages while the data should be sent via the local edges, which \emph{necessitates the computation of a routing scheme for the local network}.

The simplest solution to compute a routing scheme would be to use the global mode to collect all local device connections and/or positions in a centralized server and do the computation there. However, a centralized solution would represent a bottleneck and single point of failure, or it would not be for free when making use of a cloud service. We avoid these problems by only relying on the devices themselves. Interestingly, even without any central service, we vastly improve the results over what is possible with just the local network. %
More specifically, we demonstrate that with a hybrid wireless network one can significantly speed up the computation of compact routing schemes under certain natural circumstances, thereby opening up a new research direction for wireless networks.

\subsection{Model and Problem Definition}
\label{subsec:model+definitions}

We assume a set $V$ of $n$ nodes with unique IDs. Each node is associated with a fixed, distinct point in the 2-dimensional Euclidean plane, (i.e., $V \subseteq \mathbb R^2$), and every node $v \in V$ knows the global coordinates of its point.
We assume the standard \emph{synchronous message passing model}: time proceeds in synchronous time slots called \emph{rounds}. In each round, every node can perform an arbitrary amount of local computation and then communicate with other nodes.

In the \HYBRID model, communication occurs in one of two modes: the \textit{local mode} and the \textit{global mode}. The connections for the local mode are given by a fixed graph. In our case, this graph is represented by a \emph{unit-disk graph} $\UDG(V)$: for any $v,w \in V$, $\{v,w\} \in \UDG(V)$ if and only if $v$ and $w$ are at distance at most 1.
For the local mode, we use the \CONGEST model for simplicity: 
in each round, for all edges $\{v,w\} \in \UDG(V)$, node $v$ can send a message of $\bigO(\log n)$ bits to node $w$. However, our algorithms still work if instead the more restrictive (and more natural) \BCONGEST model is used (see Appendix~\ref{sec:works-in-broadcast}). We assume that each message can carry a constant number of node locations (this is analogous to the \textit{Real RAM model}, a standard model of sequential computation). 

For the global mode, we are using a variant of the \emph{node-capacitated clique (\NCC)} model called $\NCC_0$  \cite{ACCPSS20} that captures key aspects of overlay networks. In this model, any node $u$ can send messages to any other node $v \in V$ provided $u$ knows $v$’s ID. 
Initially, the set of IDs known to each node is just limited to its neighbors in $\UDG(V)$. Each node is limited to sending $\bigO(\log n)$ messages of $\bigO(\log n)$ bits via the global mode in each round. W.l.o.g., we assume that whenever a node $v$ knows the ID of some node $w$, it also knows $w$'s location (since this can be sent together with its ID).

\subparagraph{Model Motivation.}
The assumption that the nodes know their global coordinates is motivated by the fact that smartphones can nowadays accurately determine their location using GPS or wireless access point or base station information. However, it would also be sufficient if the nodes can determine the distance and relative angles to their neighbors in the UDG (this can be obtained via kown localization methods~\cite{HoflingerB0BSRS18}), though
with some precision loss.

The use of the $\NCC_0$ model in the global communication mode is motivated by the fact that nodes can communicate with any other node in the world via the cellular infrastructure given its ID (e.g., its phone number). Note that $\NCC_0$ is weaker than $\NCC$, which assumes a clique from the start, but it is known that once the right topology has been set up in $\NCC_0$ (which can be done in $O(\log n)$ rounds \cite{GHSW20}), any communication round in $\NCC$ can be simulated by $O(\log n)$ communication rounds in $\NCC_0$ \cite{LeightonMRR94}.

\subparagraph{Problem Definition.}
Our goal is to compute a \textit{compact routing scheme} for $\UDG(V)$, in hybrid networks where $\UDG(V)$ is connected and does not contain radio holes. $\UDG(V)$ is said to contain a \emph{radio hole} if, roughly speaking, there is an internal cycle in $\UDG(V)$ that cannot be triangulated. A precise definition will be given in Section~\ref{sec:gridgraph}.

Let $\mathcal{G}$ be a class of graphs. A stateless\footnote{In a stateless routing scheme a packet can not accumulate information along the routing path and is thus oblivious to the routing path that the packet took so far (as opposed to stateful routing)} routing scheme $\mathcal{R}$ for $\mathcal{G}$ is a family of labeling functions $\ell_G: V(G) \rightarrow \{0,1\}^+$ for each $G\in \mathcal{G}$, which assigns a bit string to every node $v$ in $G$. The label $\ell_G(v)$ serves as the address of $v$ for routing in $G$: it contains the identifier of $v$, and may also contain information about the topology of $G$.

While the identifier is given as part of the input, the label is determined in the \emph{preprocessing}.
Additionally, the preprocessing has to set up a routing function $\rho_G: V(G) \times \{0,1\}^+ \rightarrow V(G)$ for the given graph $G$ that, given the current node of a message and the label of the destination, determines the neighbor of $v$ in $G$ to forward the message to\footnote{More general definitions of routing functions exist, but we do not require the additional power afforded by stateful routing (for instance), to compute near-constant routing schemes in logarithmic time.}. A routing scheme must satisfy various properties.

First of all, it must be \emph{correct}, i.e., for every source-destination pair $(s,t)$, $\rho_G$ determines a path in $G$ leading from $s$ to $t$. Second, it must be \emph{local}, in a sense that every node $v$ can evaluate $\rho_G(v,\ell)$ locally. Third, the routing should be \emph{efficient}, i.e., the ratio of the length of the routing path and the shortest path --- also known as the \emph{stretch factor} --- should be as close to 1 as possible. In our case, the length of a routing path is simply determined by the number of edges used by it. Note that whenever we have a constant stretch w.r.t.\ the number of edges in $\UDG(V)$, we also have a constant stretch w.r.t.\ the sum of the Euclidean lengths of its edges, so we achieve a constant stretch for \emph{both} types of metrics (see Section~\ref{sec:gridgraph}). Finally, the routing scheme should be \emph{compact}, i.e., the labels $\ell_G(v)$ of the nodes $v$ and the amount of space needed at each node $v$ to evaluate $\rho_G(v,\ell)$ should be as small as possible.

\subparagraph*{Problem Motivation.}

There are various reasons for developing fast distributed algorithms for compact routing schemes in hybrid wireless networks. First of all, computing routing schemes for the local ad-hoc network is useful even in the presence of a cellular infrastructure since ad-hoc connections are comparatively cheap to use and typically offer a much larger bandwidth. Also, the ability to quickly compute compact routing schemes allows for frequent adaption in case of topological changes in the wireless ad-hoc network with low overhead.

\subsection{Our Contributions}

In Appendix~\ref{sec:lowerbound} we show that it is impossible to set up a compact routing scheme with constant stretch in time $o(\sqrt{n})$ when just relying on the UDG for communication even if the geometric location of all nodes is known and the UDG is hole-free.
This poses the question of whether limited global communication can overcome this. We answer this question by showing the following result, which demonstrates the impact that a modest amount of global communication has when applied to problems which are challenging to solve locally.

\begin{theorem}
\label{thm:main}
For a \HYBRID network with a hole-free $\UDG(V)$, a compact, stateless routing scheme can be deterministically computed for $\UDG(V)$ in $\bigO(\log n)$ rounds. The scheme uses node labels of $\bigO(\log n)$ bits and a mapping $\rho$ that (i) can be evaluated locally with $\bigO(\log n)$ bits of information in each node and (ii) such that for every source-destination pair $s, t \in V$, $\rho$ determines a routing path of constant stretch from $s$ to $t$ in $\UDG(V)$.
\end{theorem}

Technical novelties of this work include a \emph{grid graph abstraction} of any UDG which serves to sparsify the UDG while preserving its geometric structure.
Computations on the grid graph can be simulated efficiently in $\UDG(V)$.
Furthermore, we can transform a routing scheme on the grid graph to one in the UDG, increasing the stretch only by a constant.

We show how to construct this abstraction in a distributed setting based entirely on local communication. This could potentially make it of interest when studying routing or distance approximation problems on UDGs in the \CONGEST or \BCONGEST models or for simplification of existing algorithms.
We also believe that the grid graph abstraction and its properties will be useful for future work in the \HYBRID setting, making it a springboard for the case of UDGs \emph{with} radio-holes.

\subsection{Overview}
The first step is the computation of a simple, yet surprisingly useful abstract graph structure on $\UDG(V)$, which we call a \textit{grid graph} $\Gamma$. The vertices of $\Gamma$ are the centers of the cells of a regular square grid which intersect with an edge of $\UDG(V)$. Two vertices of $\Gamma$ share an edge iff their cells are vertically or horizontally adjacent (see Section \ref{sec:gridgraphdef}, Figure \ref{fig:grid-polygon}).
Subsequently, in Section \ref{sec:active-nodes}, we tie the graphs $\UDG(V)$ and $\Gamma$ together by defining a \textit{representative} in $V$ for each vertex of $\Gamma$ that fulfills two main properties. First, two representatives of \textit{adjacent} grid vertices are connected with a path of at most 3 hops in $\UDG(V)$ (see Figure \ref{fig:representative-edge}). Second, each node in $V$ has such a representative within 1 hop in $\UDG(V)$.

We then turn to the algorithmic aspects of $\Gamma$. In Section \ref{sec:gridrep} we define the \textit{representation} $R$ of $\Gamma$, where grid vertices correspond to their aforementioned representatives and grid edges correspond to paths of 3 hops in $\UDG(V)$, and we show that $R$ can be efficiently computed in $\UDG(V)$. Furthermore, the representation $R$ can be used to efficiently \textit{simulate} the \HYBRID model on $\Gamma$, which is summarized in Theorem \ref{thm:hybrid_on_grid}.
In Section \ref{sec:constant_stretch}, we show that an optimal path in $\Gamma$ implies a path in $R$ with a constant approximation ratio (Theorem \ref{thm:grid-routing-guarantees}).

The final step of the first part is to construct a constant stretch routing scheme $\mathcal R$ for $\UDG(V)$ assuming that we have an optimal one for $\Gamma$ (Section \ref{sec:constant_stretch_routing_scheme}), which is encapsulated by Theorem \ref{thm:routing_scheme_udg}.
Since we can efficiently simulate the \HYBRID model on $\Gamma$ (Theorem \ref{thm:hybrid_on_grid}), the second part can be considered in isolation from the first part. Note that so far we did not exploit the fact that $\UDG(V)$ is hole-free. In fact, the construction, simulation, and properties of $\Gamma$ hold without that assumption, which is only needed for the second part.

Requiring the UDG to be hole-free is a strong assumption. However, we believe that at least bounding the number of holes is necessary in order to compute a compact, constant-stretch labelling scheme in $\bigO(\log n)$ rounds. Doing this in time and space polylogarithmic in $n$ that also \textit{scales well} in the number of radio holes in the UDG seems to be highly non-trivial, as these holes may intertwine in arbitrary ways, while there are exponentially many possibilities of navigating around them.\footnote{The number of simple $st$-paths that cannot be continuously deformed into each other without crossing a hole (i.e., non-homotopic paths) is $2^h$, where $h$ is the number of radio holes.} While in our setting there still can be exponentially many simple paths between two points, we are able to exploit the lack of large holes between them to deal with arbitrarily complex boundaries of UDGs in a hybrid network setting.

To compute the routing scheme $\mathcal R_\Gamma$ on $\Gamma$, the first step (Section \ref{sec:portals_and_labeling}) is to arrange the grid nodes into maximal vertical lines, called \emph{portals} (see Figure \ref{fig:grid_preprocessing:a}). All portals with two horizontally adjacent nodes will add one such edge, resulting in a \textit{portal-tree} $T_\Gamma$, which is cycle-free because $\UDG(V)$ is hole-free (see Figure \ref{fig:grid_preprocessing:b}).
In order to compute a labelling scheme we first perform a distributed depth-first traversal on $T_{\Gamma}$ (where the root is the node with min ID). This allows us to compute intervals $I_v$ for each node $v$ of $T_\Gamma$ that fulfill the parenthesis theorem: it is $I_w \supset I_v$ ($I_w \subset I_v$) for each ancestor (descendant) node $w$ of $v$ in $T_\Gamma$, or else $I_w \cap I_v = \emptyset$ when $v,w$ are in different branches of $T_\Gamma$ (see Figure \ref{fig:grid_preprocessing:d}).
Then all nodes of a portal will agree on interval $I_r$ of node $r$ that is closest to the root as their \textit{portal label}. The challenge here is to carefully line up techniques for the more restrictive $\NCC_0$ model to obtain such a labelling in $\bigO(\log n)$ rounds.

Finally, in Section~\ref{sec:grid_routing}, we use $T_\Gamma$ to route a packet from source $s$ to target node $t$ in $T_\Gamma$. Since the shortest path in $T_\Gamma$ may not necessarily follow the tree, we have to define a routing strategy that jumps over branches when needed, for which we can use the ``tree information'' encoded in the labels. We use the portal labels to prioritize jumping horizontally as soon as the next portal on a path  is reachable via any edge in $\Gamma$. Vertical routing within portals is done as a second priority for which node labels $I_v$ are used.
We prove that this strategy yields an exact routing scheme $\mathcal R_\Gamma$ for $\Gamma$ formalized in Theorem \ref{thm:routing_scheme_grid}. Consequently, Theorem \ref{thm:main} is a corollary from the fact that we can emulate $\Gamma$ on $\UDG(V)$ (Theorem \ref{thm:hybrid_on_grid}) and that $\mathcal R_\Gamma$ can be transformed into a constant stretch routing scheme $\mathcal R$ for $\UDG(V)$ (Theorem~\ref{thm:routing_scheme_udg}).

\subsection{Related Work}
\label{sec:related_work}

An early effort to formalize hybrid communication networks by \cite{ALSY90}, combined the \LOCAL model with a global communication mode that essentially allows a single node to broadcast a message to all others per round. Note that this conception of the global network is fundamentally different to ours, which manifests in the fact that solving a aggregations problem (e.g., computing the sum of inputs of each node) can take $\Omega(n)$ rounds (by contrast, it takes $\bigO(\log n)$ rounds in the \NCC model).

Recently, shortest path problems in general \textbf{hybrid networks} have been studied by various authors \cite{AHK+20,CLP20,KuhnS20,FHS20}, which provide approximate and exact solutions for the all-pairs shortest paths problem (APSP) and the single-source shortest paths problem (SSSP). These solutions all require $O(n^\varepsilon)$ rounds (for constant $\varepsilon>0$) to achieve a constant approximation ratio, and this is tight in the case of APSP. $\bigO(\log n)$-time algorithms to solve SSSP for some classes of sparse graphs (not including UDGs) are given in \cite{FHS20}. 
Shortest path problems have also been studied for hybrid wireless networks \cite{CastenowKS20}. They show that for a bounded-degree $\UDG(V)$ with a convex outer boundary, where the bounding boxes of the radio holes do not overlap, one can compute an abstraction of $\UDG(V)$ in $\bigO(\log^2 n)$ time so that paths of constant stretch between all source-destination pairs {\em outside of} the bounding boxes can be found (a simple extension of their approach to outer boundaries of {\em arbitrary} shape seems unlikely).

Numerous \textbf{online routing strategies} have been proposed for general UDGs, including FACE-I, FACE-II, AFR, OAFR, GOAFR and GOAFR+ \cite{BoseMSU01,KuhnWZ03,KuhnWZZ03,KuhnWZ02}. In \cite{KuhnWZ03,KuhnWZZ03} it is proven that GOAFR and GOAFR+ are asymptotically optimal w.r.t.\ path length compared to any geometric routing strategy. However, the achieved stretch is linear in the length of a shortest path. When a UDG contains the Delaunay graph of its nodes, one can exploit the fact that the Delaunay graph is a 2-spanner of the Euclidean metric \cite{Xia13}, and MixedChordArc has been shown to be a constant-competitive routing strategy for Delaunay graphs \cite{BonichonBCDHS18}. This is only applicable in UDGs where the line segment connecting two nodes of the UDG does not intersect a boundary, which is the case if it has a convex outer boundary \textit{and} is hole-free.

\textbf{Centralized constructions}\footnote{Note that in this paper, we allow ourselves just $O(\log n)$ rounds for pre-computation and each node can learn only $\polylog n$ bits per round given that it has small ($\polylog n$) degree, which can be true for every node. The local network has size $\Omega(n)$, meaning no single node can learn it completely. This inhibits solving the problem locally at some node, i.e., by \textit{direct} use of some centralized algorithm.} for compact routing schemes have been heavily investigated {for general graphs} (see, e.g., \cite{TZ01}) as well as UDGs. Here, we just focus on UDGs. Bruck et al.\ \cite{BruckGJ07} present a medial axis based naming and routing protocol that does not require geographical locations, makes routing decisions locally, and achieves good load balancing. The routing paths seem near-optimal in simulations, but no rigorous results are given.
Gao and Goswami \cite{GaoG15} propose a routing algorithm that achieves a constant approximation ratio for load balanced routing in a UDG of arbitrary shape, but the question of near-optimal routing paths is not addressed.
Based on work by Gupta et al. \cite{GuptaKR04} for planar graphs, Yan et al.\ \cite{YanXD12} show how to assign a label of $\bigO(\log^2 n)$ bits to each node of the graph such that given the labels of a source $s$ and of a target $t$, one can locally compute a path from $s$ to $t$ with constant stretch.
Using the well-separated pair decomposition (WSPD) for UDGs \cite{GaoZ05}, Kaplan et al. \cite{KaplanMRS18} present a local routing scheme with stretch $1\!+\!\eps$ with node labels, routing tables and headers of size polynomial in $\log n, \log D$, and $1/\varepsilon$, where $D$ is the diameter of $\UDG(V)$. 
Later, \cite{mulzerLIPIcs2020} shows how to achieve a stretch of $1\!+\!\eps$ without using dynamic headers.

Our routing scheme for the grid graph abstraction extends the routing scheme proposed by Santoro and Khatib~\cite{SK85}, who presented a labelling along with an optimal routing scheme for trees by computing a minimum-distance spanning tree and labelling of that tree via a depth-first search.\footnote{
While the routing scheme in \cite{SK85} guarantees a 2-approximation for general graphs regarding the worst-case optimal cost when routing over all possible source-target-pairs, their scheme does not guarantee constant stretch when routing a message between two specific nodes $s,t$ in the grid graph.}
In our scheme, we provide optimal paths between any source-target pair in the grid graph, because we allow using edges that are not part of the spanning tree for routing in order to jump between the branches of the spanning tree.

Our study is also related to routing problems in sparse graphs in \textbf{parallel models} \cite{kavvadias1996hammock, DPZ91}.
For example, the algorithm of Kavvadias et al. \cite{kavvadias1996hammock} can be used to compute routing tables in planar graphs in time $\tilO(1)$ and work $\tilO(n)$.
Together with the simulation framework of Feldmann et al.~\cite{FHS20}, the algorithm could in principle be used to solve our problem.
However, for the simulation to work, one would need to construct a suitable global network, sparsify the graph, and, together with the simulation overhead, one would obtain a polylogarithmic runtime much higher than $\bigO(\log n)$.
Further, the size of the routing tables may be $\Theta(n^2)$.

\section{Grid Graph}
\label{sec:gridgraph}

Let $G : = \UDG(V)$. The goal of this section is to construct a grid abstraction of $G$ which makes finding routing protocols in the subsequent section manageable. In particular (but still suppressing some details), we want to simulate a bounded degree \textit{grid graph} on $G$ such that shortest paths in the grid graph represents only a constant factor detour in $G$. The way we obtain such a grid representation of $G$ in a distributed fashion is by simulating grid nodes with real nodes of $V$ that are close by, where edges between grid nodes correspond to paths of constant length in $G$. We start by introducing some notations we require in the following.

\subsection{Preliminaries}
\label{sec:preliminaries}

\subparagraph{Graphs and Polygons in $\mathbb{R}^2$.}
Since each node in $V$ is associated with a point in $\mathbb{R}^2$, we can associate each edge $\{u,v\}\in \UDG(V)$ with the line segment with endpoints $u$ and $v$, i.e., the set $\{x\cdot u + (1 \m x) \cdot v \mid x \in [0,1]\}$. We use the names of vertices and edges to refer to their associated subsets of $\mathbb{R}^2$ when no ambiguity arises. 

A \emph{polygonal chain} is a finite sequence of points where consecutive points are connected by segments. A polygonal chain is \emph{closed} if the first point in the sequence is equal to the last. A \emph{polygon} is a closed, connected, and bounded region in $\mathbb{R}^2$ where the boundary consists of a finite number of (not necessarily disjoint) closed polygonal chains (this implies the edges in these polygonal chains have no proper intersections).

A \emph{hole} of a polygon $P$ is an open region in $\mathbb{R}^2$ that is a maximal bounded and connected component of $\mathbb{R}^2\setminus P$. Note that the boundary of each hole of $P$ is equal to one of the polygonal chains bounding $P$.
A polygon is \emph{simple} if it has no holes.

\subparagraph{Distance Metrics.} We use the notation $\Vert \cdot \Vert$ for the Euclidean metric on $\mathbb R^2$. Consequently, for $p, q \in \mathbb R^2$, $\Vert p - q \Vert$ denotes the Euclidean distance from $p$ to $q$. For sets of points $A,B \subseteq \mathbb R^2$ we define the distance between those sets as $\dist(A,B) := \min_{a \in A, b \in B} \Vert a \m b\Vert$.

Let $P \subseteq \mathbb{R}^2$ be a polygon in the Euclidean plane and let $p,q \in P$. We define the \textit{geometric} distance between $p$ and $q$ in $P$, $\dist_P(p,q)$, to be the length of the shortest path between $p,q$ in $P$. Note that because $P$ is a polygon, there is a polygonal chain $\Pi = (v_1, \ldots, v_k)$ from $p$ to $q$ inside $P$ such that $\dist_P(p,q) = \sum_{i=1}^{k-1} \|v_{i+1} \m v_i\|$.

Let $G = (V,E)$ be an embedded graph. Let $\Pi \subseteq E$ be a path, i.e., a sequence of incident edges of $G$. Then we define $\dist_G(\Pi) = \sum_{(u,v) \in \Pi} \|u \m v\|$. Let $|\Pi|$ be the number of edges (or \emph{hops}) of a path $\Pi$ in $G$. The \emph{hop-distance} between two nodes $u, v \in V$ is defined as  $\hop_G(u,v) := \!\min_{\text{$u$-$v$-path } \Pi} |\Pi|$.

\subsection{Grid Graph Definition}
\label{sec:gridgraphdef}

We first give some definitions to formalize the notion of an UDG having radio-holes.
A \emph{triangle of $\UDG(V)$} is a region in $\mathbb{R}^2$ that is bounded by the edges of a $3$-cycle in $\UDG(V)$ (including both the boundary and interior of the triangle). We define the \emph{contour polygon $P$ of $\UDG(V)$} as the union of all triangles and edges of $\UDG(V)$. Since $\UDG(V)$ is connected, $P$ is indeed a polygon. We call the holes in $P$ \emph{radio-holes} of $\UDG(V)$. We say an UDG has \emph{no radio-holes} if the contour polygon of that UDG has no holes, i.e., the polygon $P$ is simple.

Next we partition the plane into an axis-parallel square grid with side-length $c=\frac{1}{10}\sqrt{15}$ and a fixed origin corresponding to origin of the coordinate system.
Note that due to knowledge of coordinates, all nodes are aware of their position relative to the grid.

Define a square grid-cell to be \emph{active} if it has a non-empty intersection with $P$. Based on this grid, we define the grid graph $\Gamma= (V_\Gamma, E_\Gamma)$, where $V_\Gamma$ has a node positioned at the center of each active cell in our grid, and we have an edge in $E_\Gamma$ between every pair of nodes of $V_\Gamma$ that lie in adjacent cells in the grid (i.e., the square cells share an edge). The grid graph $\Gamma$ will be simulated in the routing protocol. We will also define the cell graph $\Gamma'=(V_{\Gamma'},E_{\Gamma'})$ in the analysis of our protocol, but do not simulate it. We call a vertex of the grid \emph{loose} if it is a corner of exactly $2$ active cells that are not adjacent. $\Gamma'$ is composed of the boundaries of the square grid, with $V_{\Gamma'}$ the set of all corners of each active grid-cell that are not loose, with a pair of vertices in $V_{\Gamma'}$ having an edge in $E_{\Gamma'}$ if they are ends of an edge of a grid-cell.
To define the \emph{cell polygon} $P'$, first take the union of all active grid-cells. Then, for every loose vertex $v$ in the grid, remove a triangle from $P'$ at every active grid-cell incident to $v$ that is small enough to be disjoint from $P$, such that $P'$ no longer contains $v$. Note that since a loose vertex does not lie in $P$ (otherwise, all $4$ cells incident to it would have been active), such a triangle exists. 
See Figure~\ref{fig:grid-polygon} for an example of these definitions.
Next, we define a \emph{representative} $r$ for each grid node $g \in V_\Gamma$, which simulates $g$ throughout the rest of the protocol. We apply one of the following rules to assign a grid node to a node $u \in V$.

\begin{figure}
    \centering
    (a) \includegraphics[width=0.43\textwidth]{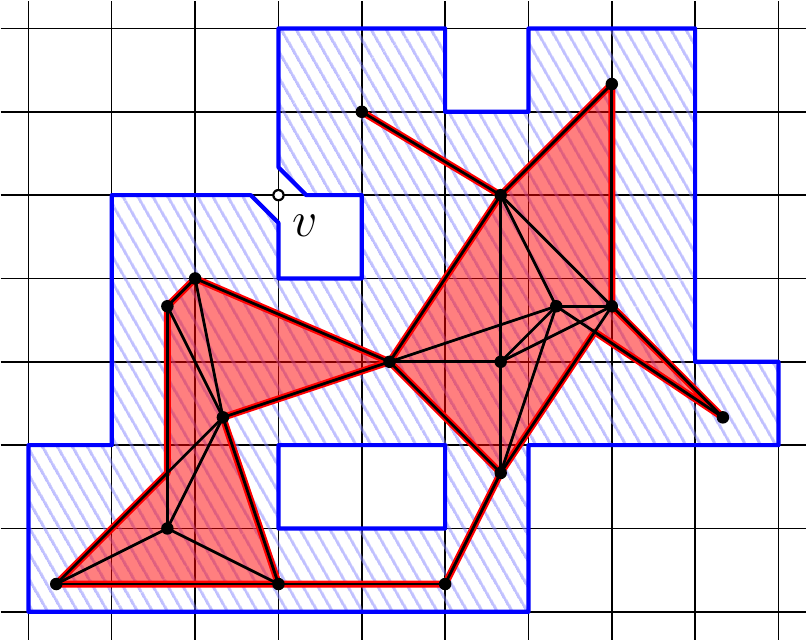}\hfil
    (b) \includegraphics[width=0.43\textwidth]{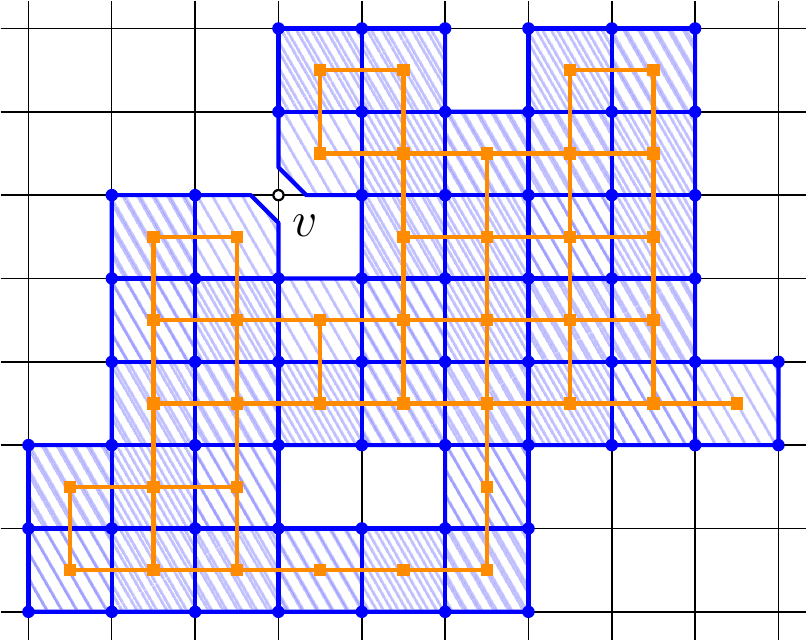}
    \caption{(a): $G := \UDG(V)$ (\textcolor[rgb]{0.00,0.00,0.00}{black}), polygon $P$ (\textcolor[rgb]{1.00,0.00,0.00}{red}), and cell polygon $P'$ (\textcolor[rgb]{0.00,0.07,1.00}{blue}). (b): grid graph $\Gamma$ (\textcolor[rgb]{1.00,0.50,0.00}{orange}), active grid-cells and cell graph $\Gamma'$ (\textcolor[rgb]{0.00,0.07,1.00}{blue}). The grid vertex $v$ is loose. Note that \boldmath $G$ and $\Gamma$ have a hole. Our routing algorithm on $\Gamma$ would not work for this UDG, but we can still construct~$\Gamma$.}
    \label{fig:grid-polygon}
\end{figure}

\begin{definition}
    \label{def:representatives}
    Let $g \in V_\Gamma$, and let $C$ be the grid cell of which $g$ is the center.
    We define $\mathcal{C}_1(g)$ as the set of vertices of all triangles of $\UDG(V)$ that contain the point $g$. We define $\mathcal{C}_2(g)$ as the set of vertices incident to an edge that intersects $C$.

    We define the \emph{set of candidate representatives} $\mathcal{C}(g)$ as $\mathcal{C}_1(g) \cup \mathcal{C}_2(g)$.

    The representative of $g$ is defined as $r = \argmin_{v\in \mathcal C_1(g)} \|v-g\|$ if $\mathcal{C}_1(g)$ is non-empty, and $r = \argmin_{v\in \mathcal C_2(g)} \|v-g\|$ otherwise. In either case, we break ties by smallest node ID.
\end{definition}

\subsection{Properties of the Grid Graph}
\label{sec:active-nodes}

The next step is to show that the grid abstraction introduced in Definition \ref{def:representatives} represents the UDG well. In this section we prove several properties to this effect: we show that nodes are adjacent to the representative of the cell which they are in (Lemma~\ref{lem:hops_node_rep}); that representatives for adjacent grid cells are close (Lemma~\ref{lem:rep_grid_3_hops}); and that the cell polygon $P'$ is simple (Lemma~\ref{lem:cell-polygon-simple}).

\begin{lemma}
	\label{lem:hops_node_rep}
	Let $u,r\in V$. If $r$ is the representative of the cell $C$ containing $u$, then $\hop_G(u,r)\leq 1$
\end{lemma}

Intuitively, this is true because $u$ must be close to the centre of $C$, as must $r$: even if these nodes are different they cannot be too far apart.

\begin{proof}
	Let $g$ be the center of $C$. Because $u$ lies in $C$, we have $u\in \mathcal{C}(g)$. If $r\in \mathcal{C}_2(g)$ or $u\notin \mathcal{C}_2(g)$, we have $\|g-r\|\leq \|g-u\|$, because otherwise $r$ would not be the representative of $g$. Together with the triangle inequality, we get $\|r-u\| \leq \|g-r\| + \|g-u\|\leq 2\|g-u\|\leq \sqrt{2}\cdot c =\frac{1}{10}\sqrt{30}\leq 1$.
	
	Otherwise, we have $r\in \mathcal{C}_1(g)$ and $u\in \mathcal{C}_2(g)$. This means $g$ lies inside a triangle of $G$ and $u$ does not. So, there is an edge $e$ of a triangle of $G$ that intersects $C$. Let $g'$ be the closest point to $g$ on $e$. We have $\|g'-g\|\leq \frac{1}{2}\sqrt{2} \cdot c$ because $e$ intersects $C$. There is an endpoint $v$ of $e$ with $\|v-g'\|\leq \frac{1}{2}$, because the length of $e$ is at most $1$. Since $gg'$ is orthogonal to $e$, we have $\|v-g\| = \sqrt{\|v-g'\|^2 + \|g-g'\|^2}$. Therefore, we have $\|r-u\| \leq \|g-r\| + \|g-u\| \leq \|v-g\| + \|g-u\| \leq \sqrt{\frac{1}{4} + \frac{1}{2}c^2} + \frac{1}{2}\sqrt{2} \cdot c \leq 1$.
\end{proof}

Next, we show that the representatives of adjacent grid cells are within $3$ hops of each other.

\begin{lemma}
	\label{lem:rep_grid_3_hops}
	Let $(g_1, g_2) \!\in\! E_\Gamma$ be an edge in $\Gamma$. Let $u, v \!\in\! V$ be representatives of $g_1, g_2$ respectively. Then $\hop_G(u,v) \!\leq\! 3$.
\end{lemma}

\begin{proof}
	We consider two cases: First, suppose either $u$ or $v$ is a vertex of a triangle that contains both $g_1$ and $g_2$. Then $u$ and $v$ are both candidates for the same grid node, so $\hop_G(u,v)\leq 3$ by Lemma~\ref{lem:rep-candidates-3-hops}.
	
	Otherwise, both $u$ and $v$ are either a vertex of a triangle of which the boundary intersects the segment $g_1g_2$ (because the triangle contains exactly one of two grid nodes) or an endpoint of an edge that intersects the cell of $g_1$ or $g_2$. In both cases, there exists an edge $e(u)$ that (i) intersects the union of the cells belonging to $g_1$ and $g_2$, and (ii) has both ends within hop-distance at most $1$ to $u$. Analogously, there exists a segment $e(v)$ with the same properties.
	
	Since $e(u)$ and $e(v)$ both intersect a rectangle of size $2c$ by $c$, their distance is at most $\sqrt{(2c)^2 + c^2}=\sqrt{5}\cdot c =\frac{1}{2}\sqrt{3}$. By Lemma~\ref{lem:close-segments-connect}, one end of $e(u)$ is adjacent to an end of $e(v)$. Since both ends of $e(u),e(v)$ have distance at most $1$ to $u,v$ respectively, we have $\hop_G(u,v)\leq 3$.
\end{proof}

We show that the edges which define $\mathcal C(g_1)$ and $\mathcal C(g_2)$ are at most the diagonal of a $2 \times 1$ block of grid cells apart.
We conclude that this distance is small enough that an edge connects an endpoint of one edge with an endpoint from the other, and so the representatives of adjacent cells have distance at most $3$ from each other.

Finally, we show that $P'$ is simple, i.e., it has no holes. We show this by observing that if there is a hole in $P'$, there is a cycle of active cells with an inactive cell in its interior. We show this cycle of cells contains a cycle of $G$, which implies $G$ contains a radio-hole.

\begin{lemma}
	\label{lem:cell-polygon-simple}
	If $G$ has no holes, then $P'$ is simple.
\end{lemma}

\begin{proof}
	We prove the contrapositive. Suppose $P'$ is not simple. Then, one of the polygonal chains $\Pi$ bounding $P'$ contains only inactive grid cells in its interior. Let $A$ be the set of active cells with a vertex of $\Pi$ on its boundary. Because $\Pi$ does not contain any loose vertices, $A$ forms a cycle of cells where adjacent cells share an edge. Since $\Pi$ contains only inactive cells in its interior, $A$ lies outside of $\Pi$. Each cell in $A$ shares a corner with a grid-cell that is not active, because all  vertices of $\Pi$ lie on the boundary of an inactive grid cell. So, each cell in $A$ intersects the boundary of $P$ and therefore intersects an edge in $G$. Consider two adjacent cells $C_1,C_2$ in $A$ containing edges $e_1,e_2$, respectively. These edges intersect a rectangle of size $2c$ by $c$, so their distance is at most $\sqrt{(2c)^2 + c^2}=\sqrt{5}\cdot c =\frac{1}{2}\sqrt{3}$. By Lemma~\ref{lem:close-segments-connect}, one end of $e_1$ is adjacent to an end of $e_2$ in $G$. Since $A$ lies outside $\Pi$ there is a cycle in $G$ that has an inactive cell in its interior. This means the cycle bounds a region of $\mathbb{R}^2\setminus P$. But then the polygon $P$ of $G$ is not simple, so $G$ has a radio-hole.
\end{proof}

\subsection{Grid Graph Representation, Computation and Simulation}
\label{sec:gridrep}

Building on the previous subsections, we show that we can efficiently simulate the grid graph $\Gamma$ with a sub-graph $R = (V_R,E_R)$ of the UDG $G$ which we call a \textit{representation} of $\Gamma$ in $G$ which closely approximates the structure of $\Gamma$.
In a nutshell: the set of nodes $V_R$ contains the set of representatives of all grid nodes $V_\Gamma$. On top of that, for each grid edge in $E_\Gamma$, we add a path in the UDG $G$ to $R$ between two representatives of the corresponding grid nodes (see example in Figure \ref{fig:representative-edge}). Note that in the previous subsection we have shown the existence of such paths that have at most 3 hops.

\begin{figure}
    \centering
    \includegraphics[height=2.3cm]{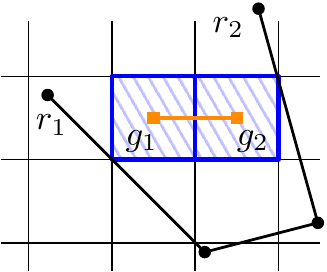}
    \caption{Representatives $r_1,r_2$ of adjacent grid nodes $g_1,g_2$ are connected by a path of $3$ hops.}
    \label{fig:representative-edge}
\end{figure}

The first goal of this subsection is to thoroughly define $R$ and to show that we can compute $R$ in $G$ according to that definition in $\bigO(1)$ rounds.
The second goal is to give an interfacing theorem for later sections that purely work with $\Gamma$, showing that a round of \HYBRID in the grid graph $\Gamma$ can be simulated in $\bigO(1)$ rounds by the nodes in $R$. By \emph{simulation}, we mean that one round of local communication between adjacent grid nodes in $\Gamma$ can be performed using $\bigO(1)$ rounds of local communication in $G$ to route messages between the representatives of adjacent grid nodes. An analogous property holds for the global communication.

\begin{definition}
    \label{def:grid_graph_represenation}
    Let $\Gamma = (V_\Gamma, E_\Gamma)$ be the grid graph as defined in Section \ref{sec:gridgraph}. A \emph{representation} $R = (V_R,E_R)$ of $\Gamma$ in $G$ is a sub-graph of $G$ defined as follows. For every grid node $g \in V_\Gamma$ with representative $r \in V$ we define: $r \in V_R$. For each edge $\{g_1,g_2\} \in E_\Gamma$ let $r_1,r_2 \in V$ be the corresponding representatives. Then $R$ contains all nodes and edges of one $r_1$-$r_2$-path $\Pi_{r_1,r_2}$ in $G$ such that $|\Pi_{r_1,r_2}| \leq 3$. We call $\Pi_{r_1,r_2}$ the representation of the edge $\{g_1,g_2\}$. Note that such a path always exists due to Lemma \ref{lem:rep_grid_3_hops}.
\end{definition}

\begin{lemma}
    \label{lem:compute_representation}
    A representation $R=(V_R,E_R)$ of $\Gamma$ can be computed in $\bigO(1)$ rounds.
\end{lemma}
\begin{proof}
    We first compute the representatives of all grid nodes $g \in V_\Gamma$ in $\bigO(1)$ rounds (see Lemma \ref{lem:decide_representatives}). The representatives $r_g, g \in V_\Gamma$ will serve as source nodes to run all the $\mathcal B_{r_g}$ in parallel. From Lemma~\ref{lem:candidates_are_close_to_grid_nodes}, we have $\|r_g, g\| \leq (1 + \frac{c}{\sqrt{2}})$. By construction of $\mathcal B_{r_g}$, only nodes that are within distance 3 of $r_g$ can ever receive a message from $\mathcal B_{r_g}$. Hence only nodes within distance $\big(4 + \frac{c}{\sqrt{2}}\big)$ of $g$ can participate in $\mathcal B_{r_g}$. Thus, by Lemma \ref{lem:parallelize-algorithms} we can run the $\mathcal B_{r_g}$ in parallel.

    Due to Lemma \ref{lem:rep_grid_3_hops} we have $\hop_G(r,r') \leq 3$ for any pair of representatives $r,r'$ whose corresponding grid nodes $g,g'$ are adjacent in $\Gamma$. Hence we will add exactly one path $\Pi_{r,r'}$ (including all nodes on $\Pi_{r,r'}$) to $R$; either the branch in $T_{r'}$ in which $r$ is located in case ID($r$) $>$ ID($r'$), or vice versa.
\end{proof}

\begin{theorem}
    \label{thm:hybrid_on_grid}
    A round of the \HYBRID model in $\Gamma$ can be simulated in $\bigO(1)$ rounds.
\end{theorem}
\begin{proof}
    We simulate $\Gamma$ with a representation $R$ of $\Gamma$. Lemma \ref{lem:compute_representation} shows that the construction of $R$ can be done in $\bigO(1)$ rounds. The representative $r \in V_R$ of some grid node $g \in V_\Gamma$ is responsible for simulating $g$. If a given representative has to simulate multiple grid nodes $g_1, \dots, g_c$, where $c \in \bigO(1)$ by Corollary \ref{cor:constant_responsibility}, then it assigns the simulated grid node $g_i$ the identifier $ID(g_i) := i \circ  ID(r)$ (where ``$\circ$'' denotes the concatenation of strings and $ID(r)$ has length at least $\lceil \log n \rceil$ with leading zeros if necessary). Since $r$'s ID is unique in $G$, the IDs of $g_1, \dots, g_c$ are unique as well.
    The simulation of the \CONGEST model by the representatives is then covered by Lemma \ref{lem:forward_message_in_R}.

    As a consequence of the Lemma \ref{lem:forward_message_in_R}, representatives of adjacent grid nodes also know each others ID, which fulfills the formal requirement of the $\NCC_0$ model that is used as the global network. By Corollary~\ref{cor:constant_responsibility} each node has only a constant number of grid nodes it needs to simulate so we are also able to simulate a round in the $\NCC_0$ model by either combining messages or via time-multiplexing.
\end{proof}

\section{Constant Stretch Routing Scheme for the UDG}
\label{sec:constant_stretch}

It remains to show how to leverage the grid graph constructed in the previous section for the computation of routing schemes for the UDG assuming that an exact routing scheme for the grid graph is known. We start with the analysis of the approximation factor.

\subsection{From Shortest Paths in $\Gamma$ to Approximate Paths in $G$}

\label{sec:constant_stretch_paths}

The goal of this subsection is to show that shortest paths in the simulated grid graph $\Gamma$ represent good paths in the UDG $G$. In particular, paths in $G$ that are obtained via the representation $R$ of $\Gamma$ are constant approximations of \textit{optimal} paths in $G$, both in terms of hop-length and Euclidean distance. We start by defining a \textit{representative path}.

\begin{definition}
    \label{def:representative_path}
    Let $s,t \in V$. Let $g_s, g_t \in V_\Gamma$ be the two grid nodes which are located in the same grid cell as $s,t$ respectively. Let $r_s,r_t$ be the representatives of $g_s$ and $g_t$. Note that $\{s,r_s\},\{r_t,t\} \in E_G$ due to Lemma \ref{lem:hops_node_rep}. Consider an optimal $g_s$-$g_t$-path $\Pi^*$. For $e \in \Pi^*$ let $\Pi_e$ be the representation of the grid edge $e \in E_\Gamma$ (see Definition \ref{def:grid_graph_represenation}). Let $\Pi_{r_s,r_t} := \bigcup_{e \in \Pi^*} \Pi_e$. We define the \emph{representative $s$-$t$-path} as $\Pi_{s,t} := \{\{s,r_s\}\} \cup  \Pi_{r_s,r_t} \cup \{\{r_t,t\}\}.$
\end{definition}

We will show that our routing scheme routes packets from $s$ to $t$ along the representative path $s$-$t$-path $\Pi_{s,t}$. First, we show that these paths achieve constant stretch in~$G$.

\begin{theorem}
    \label{thm:grid-routing-guarantees}
    Let $s,t \in V$. Let $\Pi_{s,t}$ be the $s$-$t$-path given in Def.\ \ref{def:representative_path}. If $\{s,t\} \notin E_G$ Then
    $\dist(\Pi_{s,t}) \leq |\Pi_{s,t}| \leq 36 \cdot  \dist_G(s,t)$.
\end{theorem}

Note that if $\{s, t\} \in E_G$ then we can send the packet directly along this edge and the distance and number of hops is guaranteed to be optimal. If $\{s, t\} \not \in E_G$ then $dist_G(s, t) > 1$, a fact which we use in the proof of Theorem~\ref{thm:grid-routing-guarantees}.
We prove this theorem in stages represented by the subsequent lemmas. In the first stage we upper bound the number of hops of the representative path $\Pi_{s,t}$ with the distance of a corresponding $g_s,g_t$-path in $\Gamma$.

\begin{lemma}
    \label{lem:dist_rep_path_grid}
    Let $s,t \in V$.
    Then
    $|\Pi_{s,t}| \leq \tfrac{3}{c} \cdot \dist_\Gamma(g_s,g_t) + 2$.
\end{lemma}
\begin{proof}
    We exploit the fact that edges in $\Gamma$ have distance at least $c$ and that $\Pi_{s,t} = \{\{s,r_s\}\} \cup  \Pi_{r_s,r_t} \cup \{\{r_t,t\}\}$ is constructed from an optimal path in $\Gamma$ (see Definition \ref{def:representative_path}). We combine this with Lemmas \ref{lem:hops_node_rep}, \ref{lem:rep_grid_3_hops} to obtain the following $ |\Pi_{s,t}| \stackrel{\text{Lem. } \ref{lem:hops_node_rep}}{\leq} |\Pi_{r_s,r_t}| + 2 \stackrel{\text{Lem. } \ref{lem:rep_grid_3_hops}}{\leq}  3 \cdot \hop_\Gamma(g_s,g_t) + 2 \leq \tfrac{3}{c} \dist_\Gamma(g_s,g_t) + 2.\qedhere
   $
\end{proof}

Since the cell-polygon $P'$ completely covers $P$ (the smallest polygon containing all edges of $\Gamma$ does not, in general), we relate paths in the grid graph $\Gamma$ to paths in the cell-graph $\Gamma'$. This allows us to relate paths in $P$ to $\Gamma$.
Note that comparisons of hop-distance in $\Gamma$ and $\Gamma'$ correspond to equal comparisons of distances, since both graphs have the same granularity $c$.

\begin{lemma}
\label{lem:bound-gamma-prime}
    Let $g_1,g_2 \in V_\Gamma$ be located in cells $C_1,C_2$, respectively. There exist nodes $g_1', g_2' \in V_{\Gamma'}$ that are corners of $C_1,C_2$ respectively, such that
    $\dist_\Gamma(g_1,g_2) \leq 2 \cdot \dist_{\Gamma'}(g_1',g_2')$.
\end{lemma}

\begin{proof}
    Choose $g_1',g_2'$ such that $\hop_G(g_1',g_2')\geq 1$.
    Let $\Pi'$ be a shortest $g_1' g_2'$-path in $\Gamma'$. Since all edges and vertices of $\Gamma'$ are part of the boundary of an active grid cell and $\Gamma'$ contains no loose vertices, there is a sequence $A$ of active grid-cells from $C_1$ to $C_2$, where consecutive cells share a side and each cell has an edge or vertex of $\Pi'$ on its boundary. There are two kinds of cells in $A$: the first kind has an edge of $\Pi'$ on its boundary, the second kind does not have an edge of $\Pi'$ on its boundary, but has a vertex of $\Pi'$ on its boundary. The number of cells of the first kind is at most $|\Pi'|$, because each edge in $\Pi'$ is adjacent to at most one cell of $A$. The number of cells of the second kind is at most $|\Pi'|+1$, because each vertex of $\Pi'$ is adjacent to at most one cell of this type (since $\Pi'$ has at least one edge.). So, $|A|\leq 2|\Pi'|+1$.

    We obtain a $g_1 g_2$-path $\Pi$ of length $|A|\!-\!1$ in $\Gamma$ from the chain $A$ by taking the vertex centered at each cell in $A$. So, we have $\hop_\Gamma(g_1,g_2)\leq |\Pi|\leq |A|\!-\!1 \leq 2|\Pi'| = 2\hop_{\Gamma'}(g_1',g_2')$. Since all edges in $\Gamma$ and $\Gamma'$ have length $c$, we have $\dist_\Gamma(g_1,g_2) \leq 2 \dist_{\Gamma'}(g_1',g_2')$.
\end{proof}

We follow up on the previous stage, and bound the distance of an optimal path in the \textit{graph} $\Gamma'$ with that of an optimal \textit{geometric} path in the \textit{polygon} $P'$. The resulting approximation factor of $\sqrt 2$ stems from a segment-wise comparison of Euclidean distance of a shortest polygonal chain in $P'$ to the Manhattan distance in the graph $\Gamma'$.

\begin{lemma}
\label{lem:optimality-root-two-factor-grid}
    Let $g_1,g_2 \in V_{\Gamma'}$.
    Then $\dist_{\Gamma'}(g_1,g_2) \leq \sqrt{2} \cdot \dist_{P'}(g_1,g_2)$.
\end{lemma}

\begin{proof}
    Let $\Pi$ be the shortest \textit{geometric} path from $g_1$ to $g_2$ in $P'$. Since $P'$ is a polygon, $\Pi$ is a polygonal chain connecting vertices $g_1 =: v_1,\ldots ,v_n := g_2$, where $v_i$ are reflex vertices (i.e., vertices with an internal angle of at least $\pi$) of $P'$. Note that by construction of $P'$, all reflex vertices of $P'$ are vertices of $\Gamma'$, so we have $v_1,\ldots ,v_n \in V_{\Gamma'}$.

    Consider one such segment $s_i$. Each point of $s_i$ lies in some gridcell belonging to $P'$, because the path $\Pi$ lies in $P'$. Therefore, there is a monotone chain of gridcells connecting $v_i$ and $v_{i+1}$. Consider the axis aligned bounding rectangle $R_i$ defined by the two opposite corners $v_i, v_{i+1} \in V_{\Gamma'}$. The width and length of $R_i$ sum up to $\|v_i - v_{i+1} \|_1$ (where $\|(x,y)\|_1 = x+y$ for some $(x,y) \in \mathbb R^2$ denotes the $L_1$-norm).

    Traversing the boundary of the monotone chain of gridcells between $v_i$ and $v_{i+1}$ in the shortest possible way represents a shortest path between $v_i$ and $v_{i+1}$ in ${\Gamma'}$.
    On one hand, the length of this path equals the sum of side-lengths of $R_i$, i.e., $\dist_{\Gamma'}(v_i,v_{i+1}) = \|v_i - v_{i+1} \|_1$. On the other hand the geometric distance equals the length of $s_i$ which is \smash{$\dist_{P'}(v_i,v_{i+1}) = \|v_i-v_{i+1}\|$}. We have
    $$\dist_{\Gamma'}(v_i,v_{i+1}) = \|v_i - v_{i+1} \|_1 \leq \sqrt{2}\cdot \|v_i - v_{i+1} \|_2 = \sqrt{2}\cdot \dist_{P'}(v_i,v_{i+1}),$$
    using the equivalence property of $L_1$ and $L_2$-norms: $\|x\|_1 \leq \sqrt 2 \|x\|_2$ for any $x \in \mathbb R^2$.
    So, for each segment $S_i$ of $\Pi$, there exists a path in ${\Gamma'}$ with stretch at most $\sqrt{2}$ connecting the endpoints. Concatenating these paths gives the required $g_1$-$g_2$-path in ${\Gamma'}$.
\end{proof}

Next we observe that an optimal path between two nodes in the UDG $G$ can not be any shorter than a corresponding shortest geometric path in $P'$.

\begin{lemma}
    \label{lem:dist_comparison_G_cell-polygon}
    Let $s,t \in V$. Then $\dist_{P'}(s, t) \leq \dist_G(s, t)$.
\end{lemma}

\begin{proof}
    Let $\Pi$ be a shortest $st$-path in $G$. By definition, each cell that is intersected by an edge of $\Pi$ is active and therefore this edge lies in $P'$. So, $\Pi$ is an $st$-path in $P'$.
\end{proof}

We now use the inequalities proven in the lemmas above to prove Theorem \ref{thm:grid-routing-guarantees}. 

\begin{proof}[Proof of Theorem \ref{thm:grid-routing-guarantees}]
    Let $s,t \in V$ and let $g_s,g_t \in \Gamma$ be their cell representatives. Let $g_s',g_t' \in V_{\Gamma'}$ be two corner-nodes of $g_s,g_t$ for which Lemma \ref{lem:bound-gamma-prime} holds. Then we get
    \begin{align*}
        |\Pi_{s,t}| & \leq \tfrac{3}{c} \cdot \dist_\Gamma(g_s,g_t) + 2 \tag*{\small\textit{Lemma \ref{lem:dist_rep_path_grid}}}\\
        & \leq  \tfrac{6}{c} \cdot \dist_{\Gamma'}(g_s',g_t') + 2 \tag*{\small\textit{Lemma \ref{lem:bound-gamma-prime}}}\\
        & \leq  \tfrac{6\sqrt{2}}{c}  \cdot \dist_{P'}(g_s',g_t') +2 \tag*{\small\textit{Lemma \ref{lem:optimality-root-two-factor-grid}}} \\
        &\leq  \tfrac{6\sqrt{2}}{c} \cdot \big(\!\! \dist_{P'}(g_s',s) + \dist_{P'}(s,t) +\dist_{P'}(t,g_t')\big) +2 \tag*{\small \textit{triangle ineq.}} \\
        &=  \tfrac{6\sqrt{2}}{c} \cdot \big( \|g_s' \m s\| + \dist_{P'}(s,t) +\|g_t' \m t\|\big) + 2 \tag*{\small\textit{$s g_s'$ and $t g_t'$ in same cell}}\\
        &\leq  \tfrac{6\sqrt{2}}{c} \cdot \big( \|g_s' \m s\| + \dist_{G}(s,t) +\|g_t' \m t\|\big) +2 \tag*{\small \textit{Lemma \ref{lem:dist_comparison_G_cell-polygon}}}\\
        & = \tfrac{6\sqrt{2}}{c} \cdot \big( \!\! \dist_G(s,t) + \sqrt{2}\cdot c\big) +2
     = \tfrac{6\sqrt{2}}{c} \cdot \dist_G(s,t) + 14
    \end{align*}

    In the equality in the fourth step we use that the segments $s g_s'$ and $t g_t'$ are both contained in a single grid cell, hence the distance in the cell-polygon equals the Euclidean distance. Since a grid cell has side length $c$, we have $\|g_s' \m s\|, \|g_t' \m t\| \leq \tfrac{1}{2}\sqrt{2} \cdot c$ in the second last step.
    As $\Pi_{s,t}$ is a path in a UDG each edge has distance at most 1, thus \[\dist_G(\Pi_{s,t}) \leq |\Pi_{s,t}| \leq \tfrac{6\sqrt{2}}{c} \cdot \dist_G(s,t) + 14 \leq 22\cdot \dist_G(s,t) + 14.\]

    Since we have a direct edge to targets with distance at most 1,
    the additive error can be accounted for by increasing the multiplicative stretch by the additive error for targets at distance more than $1$. Consequentially, we obtain $\dist(\Pi_{s,t}) \leq |\Pi_{s,t}| \leq 36 \cdot  \dist_G(s,t). \qedhere$
\end{proof}

\subsection{Transforming Routing Schemes for the Grid Graph to the UDG}
\label{sec:constant_stretch_routing_scheme}

We provide an interface to transform a routing scheme $\mathcal R_\Gamma$ for the grid graph $\Gamma$ (for which an exact routing scheme is provided in the subsequent section) into a routing scheme $\mathcal R$ for the UDG $G$ with constant stretch. The idea is to construct $\mathcal R$ from $\mathcal R_\Gamma$ using the representation $R$ of $\Gamma$ (see Definition \ref{def:grid_graph_represenation}). Theorem \ref{thm:routing_scheme_udg} provides approximation guarantees by leveraging the insights on representative paths from the previous subsection.

\begin{definition}[UDG Routing Scheme]
    \label{def:routing_scheme_udg}
    Let $\mathcal R_\Gamma$ be an \emph{exact} routing scheme for $\Gamma$ consisting of $\ell_\Gamma: V_\Gamma \to \{0,1\}^+$ and $\rho_\Gamma: V_\Gamma \times \{0,1\}^+ \to V_\Gamma$.
    Let $R = (V_R,E_R)$ be the representation of $\Gamma$ (see Def.\ \ref{def:grid_graph_represenation}). The routing scheme $\mathcal R$ for $G$ is defined on the basis of grid cells. Let $C$ be a cell with grid node $g \in V_\Gamma$ and let $r \in V_R$ be the representative of $g$. For each $v \in C$ we set $\ell_G(v) := \ell_\Gamma(g) \circ ID(v)$ (where ``$\circ$'' represents the concatenation of bit strings).
    The routing function $\rho_G$ is defined as follows. Let $v \in V_G$ be the current node and let $\ell_t := \ell_{\Gamma,t} \circ ID(t)$ be the label of the target node $t \in V$, where $\ell_{\Gamma,t}$ is the label of the representative in $t$'s cell w.r.t. $\mathcal R_\Gamma$. We assume $t \neq v$, as otherwise the packet has already arrived.
    \begin{enumerate}
        \item If $\{v,t\} \in E_G$, then we can directly deliver to $t$: $\rho_G(v, \ell_t) := t$.
        \item Else, if $v \in C \setminus V_R$ is the source we directly route to the representative of $C$: $\rho_G(v,\ell) := r$.
        \item Else, if $v=r$ \emph{is} the representative of this grid cell $C$, let $g' := \rho_\Gamma(g,\ell_{\Gamma,t})$ be the next grid node suggested by $\mathcal R_\Gamma$. Let $u$ be the first node on the path $\Pi_{\{g,g'\}} \subseteq E_R$ that represents the edge $\{g,g'\} \in E_\Gamma$. Then $\rho_G(v, \ell_t) := u$.
        \item Else, if $v \in V_R$ but $v$ is not the representative of $C$, then $v$ must be a ``transitional node'' on $\Pi_{\{g,g'\}} \in E_R$ that represents $\{g,g'\} \in E_\Gamma$. W.l.o.g. let $g' := \rho_\Gamma(g,\ell_{\Gamma,t})$ be the next grid node suggested by $\mathcal R_\Gamma$ and $u$ be the next node on $\Pi_{\{g,g'\}}$ towards $g'$. Then $\rho_G(v, \ell_t) := u$.
    \end{enumerate}
\end{definition}

\begin{theorem}
    \label{thm:routing_scheme_udg}
    Let $\mathcal R_\Gamma$ be a local, correct, exact routing scheme for $\Gamma$ with labels and local routing information of $\bigO(\log n)$ bits. Then the routing scheme $\mathcal R$ from Definition \ref{def:routing_scheme_udg} is local, correct, has constant stretch, labels and local routing information of size $\bigO(\log n)$ bits and can be computed in $\bigO(1)$ rounds.
\end{theorem}
\begin{proof}
    Given some $t \in V$, the label $\ell_t$ of $t$ is the concatenation of the label $\ell_{\Gamma,t}$ of the grid node which is in the same cell as $t$ and ID($t$). Hence, the labelling $\ell_G$ requires $\bigO(\log n)$ bits, given that the same is true for $\ell_\Gamma$. The information required to compute $\rho_G(v, \ell_t)$ is composed of the knowledge of neighbors of $v$ in $G$ which includes the representative of the cell of $v$ (due to Lemma \ref{lem:hops_node_rep}) and the information required to evaluate $\rho_\Gamma(v, \ell_{\Gamma,t})$. Since nodes know their neighbors already as part of the problem input (local network equals the routing graph), we do not regard this as additional routing information. The information to evaluate $\rho_\Gamma(v, \ell_{\Gamma,t})$ is $\bigO(1)$ bits by our presumption.

    We continue with the correctness and the stretch of a path implied by $\rho_G$. Let $s \neq t \in V$ be the current node and the target node respectively. Consider the case that $\|s-t\| \leq 1$, i.e., the nodes are adjacent. Then, according to Definition \ref{def:routing_scheme_udg} rule (1) the packet is delivered directly to $t$ which constitutes a correct and exact path.

    Consider the case that $\|s\!-\!t\| \geq 1$. Let $g_s,g_t \in V_\Gamma$ and $r_s,r_g$ be the respective grid nodes and representatives of the cells of $s,t$. Let $\Pi^*$ be the optimal $g_s$-$g_t$-path in $\Gamma$ implied by $\rho_\Gamma$. Then the path implied by $\rho_G$ equals $\Pi_{s,t} = \{\{s,r_s\}\} \cup  \Pi_{r_s,r_t} \cup \{\{r_t,t\}\}$ from Definition \ref{def:representative_path}, where $\Pi_{r_s,r_t} := \bigcup_{{\{g,g'\}} \in \Pi^*} \Pi_{\{g,g'\}}$ and $\Pi_{\{g,g'\}}$ is the representation of the grid edge ${\{g,g'\}} \in \Pi^*$. This is due to rules (2),(3) and (4).

    By Theorem \ref{thm:grid-routing-guarantees}, we have that $\dist(\Pi_{s,t}) \leq |\Pi_{s,t}| \leq 36 \cdot  \dist_G(s,t)$, implying a stretch of $\bigO(1)$.

    The runtime of pre-computing $\mathcal R_G$ amounts to that of computing a representation $R$ of $\Gamma$, which takes $\bigO(1)$ rounds due to Lemma \ref{lem:compute_representation}. Note that in all four cases of the routing function $\rho_G$ can be evaluated locally using the representation $R$ of $\Gamma$ from the pre-computation step, information about local neighbors in $G$ and (local) evaluations of $\rho_\Gamma$.
\end{proof}

\section{Computing a Labelling for the Grid Graph} \label{sec:portals_and_labeling}

This section is dedicated to computing the labelling $\ell_\Gamma: V_\Gamma \to \{0,1\}^+$ for the grid graph by first constructing a particular tree structure $T_\Gamma$ and then computing a labelling on it in $\bigO(\log n)$ rounds leveraging various \HYBRID (and in particular $\NCC_0$) model techniques. 
For the tree-labelling we use a similar approach as presented in~\cite{SK85}, but slightly adapt the labelling which later allows jumping over branches of our specifically constructed tree, facilitating an optimal routing scheme in grid graphs.
Afterwards, Section~\ref{sec:grid_routing} will deal with computing the routing function $\rho_\Gamma: V_\Gamma \times \{0,1\}^+ \to V_{\Gamma}$ leading to the routing scheme $\mathcal R_{\Gamma}$.

We assume the \HYBRID model on the grid graph $\Gamma$ that represents the network which we constructed and simulated in the previous sections (Theorem \ref{thm:hybrid_on_grid}).
The goal is to divide the grid nodes into sets of vertically connected grid nodes called \emph{portals}.
Connecting neighboring portals with a single edge gives us a spanning tree of $\Gamma$, which we call \emph{portal tree}. We then root the portal tree at the node with minimum identifier and compute a label for each grid node, leading to a well-defined labelling function $\ell_\Gamma$.
Note that we require that the cell polygon $P'$ does not contain holes (Lemma~\ref{lem:cell-polygon-simple}), as otherwise there be a cycle after connecting neighboring portals.

We first define the set of portals as follows:

\begin{definition}[Portals]
	Let $\Gamma = (V_{\Gamma}, E_{\Gamma})$ be the grid graph as constructed in the last section.
	The \emph{set of portals} are the connected components of $(V_\Gamma, E_{vert})$, where $E_{vert} \subset E_{\Gamma}$ are the vertical edges of the grid graph.
\end{definition}

\begin{figure}
    \begin{subfigure}{.23\textwidth}
        \includegraphics[page=1, width = \textwidth]{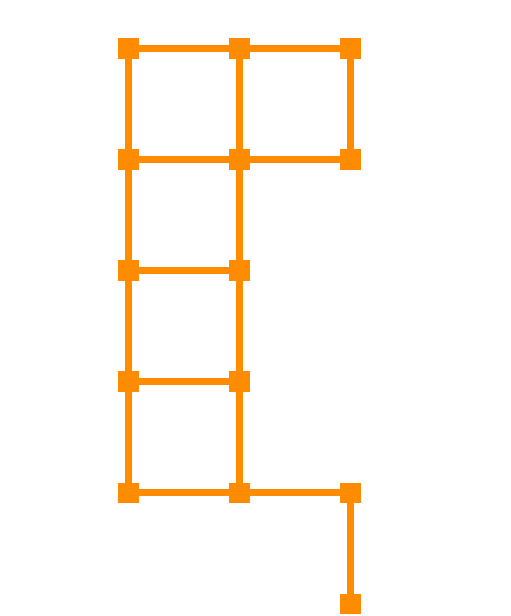}
        \caption{Initial grid graph $\Gamma$}
        \label{fig:grid_preprocessing:a}
    \end{subfigure}
    \begin{subfigure}{.23\textwidth}
        \includegraphics[page=2, width = \textwidth]{figures/grid_preprocessing-v2.pdf}
        \caption{Portal tree $T_\Gamma$ (blue)}
        \label{fig:grid_preprocessing:b}
    \end{subfigure}
    \begin{subfigure}{.23\textwidth}
        \includegraphics[page=3, width = \textwidth]{figures/grid_preprocessing-v2.pdf}
        \caption{Rooted portal tree, values $\ell_v$}
        \label{fig:grid_preprocessing:c}
    \end{subfigure}
    \begin{subfigure}{.23\textwidth}
        \includegraphics[page=4, width = \textwidth]{figures/grid_preprocessing-v2.pdf}
        \caption{Label $[\!l_v,\!r_v\!]$, portal label red}
        \label{fig:grid_preprocessing:d}
    \end{subfigure}
    \caption{Example for creation and labelling of the portal tree.}
    \label{fig:grid_preprocessing}
\end{figure}

For convenience, assume that the grid nodes $v_1,\ldots,v_k$ within a portal $\mathcal P$ are sorted by their $y$-coordinates in descending order, i.e., $v_1$ is the northernmost node.

To construct the portal tree $T_\Gamma$ of the grid graph $\Gamma$ we connect neighboring portals via a single edge.
Each grid node $v$ checks whether it has an edge to the left and communicates this to its northern and southern neighbors $v_N$ and $v_S$.
Assume that $v$ has an edge $\{v, v_W\}$ to the left.
Then $v$ checks if $v_S$ also has a horizontal edge to the left.
If that is not the case, $v$ adds the edge $\{v, v_W\}$ to the portal tree.
We refer to Figures~\ref{fig:grid_preprocessing:a} and \ref{fig:grid_preprocessing:b} for an example.

\begin{lemma} \label{lem:portal_tree_construction}
    The portal tree $T_{\Gamma}$ of a grid graph $\Gamma$ can be computed in $\bigO(1)$ rounds.
\end{lemma}

\begin{proof}
    The runtime of $\bigO(1)$ rounds is clear, as each node $v$ only needs to communicate for one round with its southern neighbor $v_S$ in the portal.
    We provide arguments on why the construction is a tree.
    Since the cell polygon $P'$ is simple (Lemma~\ref{lem:cell-polygon-simple}), the cells of vertices in a portal connect two points on the same polygonal boundary of $P'$.
    Thus, removing these cells disconnects $P'$ and therefore removing the vertices of a portal from $\Gamma$ disconnects $\Gamma$.
    This means the portal graph is acyclic, i.e., a tree.
    Since $\Gamma$ is connected, the portal graph is connected as well.
\end{proof}

Given the portal tree $T_{\Gamma}$, we want to compute a unique label for each grid node that reflects its structure as portal tree.
First, we root $T_{\Gamma}$ at the grid node $r$ whose representative is the UDG node $u$ with minimal identifier, using pointer jumping (Appendix~\ref{app:sec:pointer_jumping}) on the cycle of all grid nodes that corresponds to an Euler tour (Appendix~\ref{app:techniques:rooting}).

Now we compute the labelling for the (rooted) portal tree. For each grid node $v$ in $T_{\Gamma}$, we aim to assign an interval $I_v = [l_v, r_v] \in \mathbb{N}^2$ to $v$, such that $I_v \supset I_w$ for any child node $w$ of $v$ in $T_{\Gamma}$.
To obtain the left interval border $l_v$ for each grid node $v$ in the portal tree, we perform a depth-first traversal (DFS) on $T_{\Gamma}$ in $\bigO(\log n)$ rounds, using Lemma~\ref{app:dfs} (see Appendix~\ref{app:techniques:dfs}).
The value $l_v$ is then the preorder number of $v$ according to the DFS. 
Note that $l_v < l_u$ for any node $u$ lying in the subtree of $v$.
We then compute the number $r_v$, corresponding to the maximum left interval border among all nodes in $v$'s subtree.
In a nutshell, we first let all nodes compute some value $d \in \bigO(\log n), d \geq \log D(T_{\Gamma})$, where $D(T_{\Gamma})$ is the depth of the portal tree.
Then we generate additional edges in $T_{\Gamma}$ for $d$ iterations, by performing pointer-jumping on the paths from the leaf nodes of $T_{\Gamma}$ to the root.
We perform the pointer-jumping technique in a condensed way to ensure that the node degrees do not exceed $\bigO(\log n)$.
With the help of these additional edges, we let each node $v \in T_{\Gamma}$ compute the value $r_v$ as an aggregate of the $l_u$-values of all nodes $u$ that are contained in the subtree $T_{\Gamma}(v)$ of $T_{\Gamma}$ with root $v$.
We elaborate on this approach in Appendix~\ref{app:techniques:max_preorder} (see Lemma~\ref{lem:max_preorder_number}).

After the algorithm has terminated, each node $v$ knows the correct value $r_v$ and thus its interval $I_v=[l_v,r_v]$.
Observe that grid nodes which are in different branches of the portal tree have incomparable labels.
We obtain the following lemma:

\begin{lemma} \label{lem:labeling}
	Given a rooted portal tree $T_{\Gamma}$, each node $v \in T_{\Gamma}$ can compute an interval $I_v=[l_v,r_v]$ in $\bigO(\log n)$ rounds, such that $I_v \supset I_w$ for any child node $w$ of $v$ in $T_{\Gamma}$.
\end{lemma}

Now that each grid node $v$ knows its interval in $T_{\Gamma}$ we need to perform one final step.
In addition to its own (unique) interval, a grid node $v$ needs to know the interval that has been assigned to the node $v_i$ which is closest to the root within its own portal.
We call this label the \emph{portal label} of $v$.
The node on a portal which is closest to the root can determine this locally.
Each portal label can then be broadcasted to all nodes within the respective portal in $\bigO(\log |\mathcal P|)$ rounds (see Lemma~\ref{lemma:portal_broadcasting}), so we obtain the following lemma (cf. Figures~\ref{fig:grid_preprocessing:c} and~\ref{fig:grid_preprocessing:d}).
\begin{lemma} \label{lem:portal_label}
	After $\bigO(\log n)$ rounds, each grid node $v$ in the portal $\mathcal P=(v_1,\ldots,v_k)$ knows the interval $I_{v_i}$ of the node $i \in P$ closest to the root of the portal tree.
\end{lemma}

Observe that, the way we defined the portal labels we obtain the property that for portal labels of two neighboring portals, one portal's label is always a subset of the other.
Combining Lemma~\ref{lem:labeling} and Lemma~\ref{lem:portal_label} yields the main result of this section.

\begin{theorem}
    Computing the labelling $\ell_\Gamma: V_\Gamma \to \{0,1\}^+$ for the grid graph $\Gamma$ as part of the routing scheme $\mathcal R_{\Gamma}$ can be done within $\bigO(\log n)$ rounds.
\end{theorem}

\section{Compact Routing Scheme for the Grid Graph} \label{sec:grid_routing}
Finally, we explain our routing strategy for transmitting a packet between two nodes $s, t \in V_{\Gamma}$ in the grid graph, leading to the routing function $\rho_\Gamma: V_\Gamma \times \{0,1\}^+ \to V_{\Gamma}$.
At the start of the routing protocol, the node $s$ generates a message $m$ that contains the identifier of the target node $t$, as well as $t$'s label and portal label.
The goal of our routing strategy is to route $m$ to $t$ along grid edges via an optimal path in the grid graph.
To do so, each grid node receiving the message $m$ has to decide which of its grid neighbors to forward $m$ to, using only the information stored in $m$, and the information stored in its own local memory.
Briefly, the strategy works as follows.
While we are not at the portal containing $t$, we always try going left (west) or right (east) first by going to a portal whose label is closest to the portal label of the target node $t$.
If going east or west is not possible, we go up (north) or down (south) instead by comparing $g$'s own label with the \emph{actual} label of the target node $t$.
Once we are at the portal that contains the target node, we only consider going up or down until we reach $t$.

\subparagraph{Detailed Description.}
We describe the routing strategy in more detail now (see Algorithm~\ref{algo:routing} in Appendix~\ref{app:grid_routing}
for pseudocode).
Assume we are at a grid node $g$ and want to route a message $m$ to a grid node $t$.
We introduce the following notation for the information known to $g$. Note that grid nodes obtain this information in one communication round with their neighbors.
\begin{definition} \label{def:grid_scheme:local_information}
    The information required to be stored by a grid node $g \in V_{\Gamma}$ are denoted by the following variables.
    \begin{enumerate}[(i)]
        \item $g.L \in \mathbb{N}^2$: $g$'s own interval given to it by labelling of the portal tree.
        \item $g.P \in \mathbb{N}^2$: The portal label of the portal containing $g$.
        \item $g_N, g_S, g_E, g_W \in V_{\Gamma} \cup \{\perp\}$: $g$'s grid neighbors in north, south, east and west direction ($\perp$ denotes that there is no such neighbor). For each of these grid neighbors $g$ also knows the label of the grid node 
        and the portal label of the grid node.
    \end{enumerate}
\end{definition}

Additionally, we store the label $t.L$ of $t$ and the portal label $t.P$ of $t$ in the message $m$, so $g$ knows these as well upon receipt of $m$.
Note that storing this information at $g$ requires only $\bigO(\log n)$ bits.
Assuming that $g \neq t$, $g$ must decide which of its grid neighbors $g_W, g_E, g_N, g_S$ to forward $m$ to.
Node $g$ first checks if it is in the same portal as $t$ by comparing $g.P$ and $t.P$.
Assume that this is not the case.
Then $g$ has to consider the following cases. We use the notation $a \incomp b$ to denote that label $a$ is \emph{incomparable} to label $b$, i.e., $a \not \subseteq b \wedge b \not \subseteq a$.

We start by explaining how a message $m$ is routed in horizontal direction.
\begin{enumerate}[(i)]
    \item $g.P \subset t.P$ or $g.P \supset t.P$. In case $g.P \subset t.P$ then $g$ checks if either $g.P \subset g_W.P \subset t.P$ or $g.P \subset g_E.P \subset t.P$ holds (only one of these conditions can be true). In the first case, $g$ forwards $m$ to $g_W$, in the second case $g$ forwards $m$ to $g_E$. 
    If none of the conditions hold (for example, if $g_W = \perp$ or $g_E = \perp$), then $g$ routes $m$ vertically (see the description below). 
    The case $g.P \supset t.P$ works analogously.

    \item $g.P \incomp t.P$. 
    In this case $g$ tries to forward $m$ horizontally to a node, whose portal label is a superset of $g.P$. By doing so, $m$ eventually reaches a node $g'$ whose portal label is also a superset of $t.P$ (at the closest ``common ancestor portal''), and case (i) is considered.
    If neither $g_W$ nor $g_E$ satisfies this condition or does not exist, $g$ routes vertically.
\end{enumerate}

We now explain how $m$ is routed in vertical direction.
We do this if $g$ has not been able to route $m$ horizontally (either because its horizontal neighbors are not appropriate, or because they do not exist) or if it is already contained in the same portal as the target node $t$.
Again, $g$ considers the following cases, this time for its own label $g.L$ instead for $g.P$ and for the actual label $t.L$ instead of the portal label $t.P$.
\begin{enumerate}[(i)]
    \item $g.L \subset t.L$ or $g.L \supset t.L$. In the case $g.L \subset t.L$ node $g$ checks if either $g.L \subset g_N.L \subset t.L$ or $g.L \subset g_S.L \subset t.L$ holds. In the first case, $g$ forwards $m$ to $g_N$, in the second case $g$ forwards $m$ to $g_S$. The case $g.L \supset t.L$ works analogously.
    \item $g.L \incomp t.L$.
    If the labels $g.L$ and $t.L$ are incomparable, $g$ tries to forward $m$ vertically to a node, whose label is a superset of $g.L$.
    This is the case for either $g_N$ or $g_S$, depending on the location of the root of the labeled tree.
\end{enumerate}

\subparagraph{Analysis of the Routing Strategy.}
We show that our routing strategy is local, efficient, and correct, so it fulfills all requirements for a routing scheme.
Our routing strategy is local, as each node $v$ can determine the next node to forward the message $m$ to based solely on the $\bigO(\log n)$ bits of local information, and the labels $t.L$ and $t.P$ given to $v$ upon receipt of $m$.

Regarding efficiency of our routing strategy, we prove with arguments in Appendix~\ref{app:grid_routing} 
that it is optimal. The idea is to show that in case the message is routed in a specific direction, there exists at least one optimal path that moves in the same direction.
We conclude the following theorem.

\begin{theorem}
    \label{thm:routing_scheme_grid}
    A local, correct and exact routing scheme $\mathcal R_{\Gamma}$ for $\Gamma$ using node labels and local space of $\bigO(\log n)$ bits 
    can be computed in $\bigO(\log n)$ rounds in the \HYBRID model.
\end{theorem}

\section{Conclusion}

We showed that for any \HYBRID network with a hole-free $\UDG(V)$, a compact routing scheme can be computed for $\UDG(V)$ in just $\bigO(\log n)$ rounds. There are various interesting directions for follow-up research. For example, we suspect that our approach can be generalized to 3 dimensions (potentially more) where the corresponding ``unit ball graph'' implies a polyhedron of genus 0. In particular, some approach akin to multidimensional range trees might work: define $\Gamma$ analogously in a three dimensional grid; dissect $\Gamma$ along 2d-hyperplanes to obtain 2d-portals in $\Gamma$ --- if one then comes up with a routing scheme to find the correct 2d-portal, then this can be applied alongside the 2d-routing algorithm presented here to find the correct node in that 2d-portal. There are unresolved issues, however. Another interesting direction is to efficiently compute compact routing schemes for \emph{arbitrary} connected UDGs, or ideally, to find efficient solutions for arbitrary planar graphs. This seems to be a daunting task; a simpler setting might be to consider UDGs with a small number of holes where our grid construction could be of help. Finally, it would be interesting to think about adaptations of our routing scheme to also minimize congestion, which should be possible in the special case of hole-free UDGs (see for example the case where the contour polygon is a square~\cite{CKK12}).

\phantomsection
\addcontentsline{toc}{section}{References}
\bibliography{literature}

\newpage\appendix

\section{Lower Bound Without Global Communication} \label{sec:lowerbound}

The counterexample in Figure~\ref{fig:wheel} (which follows the arguments of \cite{KuhnWZ02}) demonstrates that it is impossible to set up a compact routing scheme with constant stretch in polylogarithmic time when just relying on the unit-disk graph, even if it does not have radio holes and the geometric location of the destination is known: Suppose the destination is in the center. With a wheel of $\Theta(\sqrt{n})$ spikes of $\Theta(\sqrt{n})$ length each, no node on the wheel can guess the right spike with probability better than $\Theta(1/\sqrt{n})$, so routing information is required for a constant stretch, but in order to compute the information needed for a constant stretch, the starting point $w$ of the spike leading to the destination in the center needs to be identified, which requires $\Omega(\sqrt{n})$ communication rounds.

\begin{theorem}
    There is no deterministic (randomized) distributed algorithm that can, within $o(\sqrt{n})$ rounds, compute a compact routing scheme that achieves (expected) $o(\sqrt{n})$ stretch when only communicating over the unit-disk graph. This claim holds even when the algorithm can use geometric information\footnote{i.e., each node knows and can communicate its own location, the location of its neighbours, and each source will be given the location of its destination before routing.} and the unit-disk graph has no radio-holes.
\end{theorem}

\vspace{-4mm}
\begin{figure}[H]
    \centering
    \includegraphics[width=0.37\textwidth]{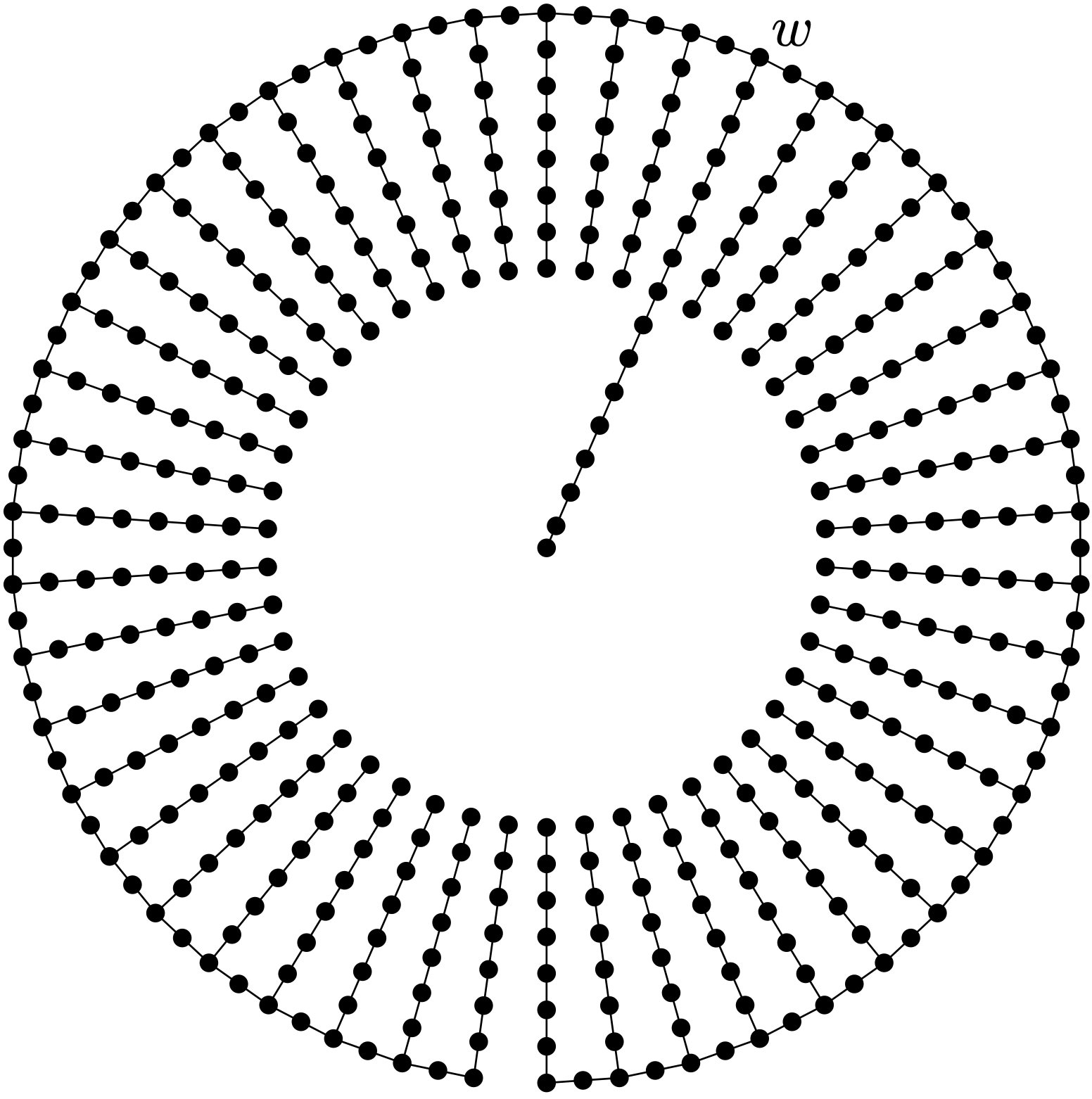}
    \caption{Lower bound graph (slightly adapted from Figure 8 in \cite{KuhnWZ02}).}
    \label{fig:wheel}
\end{figure}

\section{Additional Technical Details for Section~\ref{sec:active-nodes}}\label{app:geometry-lemmas}

We provide the full proofs for the lemmas in Section~\ref{sec:active-nodes}, as well as a few auxiliary lemmas for some geometric properties that we require when we work with unit disc graphs.

The following is the proof of the lemma that claimed that all nodes inside a grid cell $C$ are neighbours of the representative of $C$.

The following two lemmas concern an important geometric property which is required for the remaining proofs in this section.
By definition, a pair of points in an UDG is connected with an edge if the distance between them is small enough. The same is true for a pair of edges in a UDG. We formalize this idea in Lemma~\ref{lem:close-segments-connect} below.
Specifically, we show that if two edges $e_1$ and $e_2$ are within distance $\frac{1}{2} \sqrt{3}$ of each other, then at least one endpoint from $e_1$ is connected to at least one endpoint of $e_2$.

\begin{lemma}
\label{lem:points_on_segments}
    Let $s_1$, $s_2$ be two disjoint line segments. Then at least one of the endpoints of $s_1$ or $s_2$ has minimum distance exactly $\dist(s_1,s_2)$ to the respective other segment.
\end{lemma}
\begin{proof}
    The distance between segments is the shortest distance between a pair of points on the segment. If all closest pairs of points from $s_1$ and $s_2$ contain an endpoint, we are done, so assume there exists a pair of points $p_1,p_2$ that realize the shortest distance between $s_1,s_2$ and are not endpoints.
    The nearest point on a segment $s$ to a point $p$ is either the orthogonal projection $p'$ of $p$ onto the line through $s$ (if $p'$ lies on $s$) or the nearest endpoint of $s$ to $p'$. So, $p_1$ is the orthogonal projection of $p_2$ on $s_2$ and vice versa. As the segments are disjoint, we have $p_1\neq p_2$, so the line through $p_1$ and $p_2$ is orthogonal to both $s_1$ and $s_2$.

    This means $s_1$ and $s_2$ are parallel and the nearest endpoint of $s_1$ and $s_2$ to the right of $p_1,p_2$ has a projection on the other segment with the same distance.
\end{proof}

\begin{lemma}\label{lem:close-segments-connect}
Let $s_1,s_2$ be two segments of length at most $1$. If the distance between $s_1$ and $s_2$ is at most $\frac{1}{2}\sqrt{3}$, then there is an endpoint $p_1$ of $s_1$ and an endpoint $p_2$ of $s_2$ such that $\|p_1-p_2\|\leq 1$.
\end{lemma}
\begin{proof}
    If $s_1$ intersects $s_2$ in some point $x$, then at least one endpoint of both segments has distance at most $1/2$ to $x$. By the triangle inequality, these endpoints have distance at most $1$ to each other.

    Otherwise, $s_1$ and $s_2$ are disjoint. By Lemma~\ref{lem:points_on_segments}, there is an endpoint $p$ with distance at most $\frac{1}{2}\sqrt{3}$ to the other segment $s$. If the orthogonal projection of $p$ on the line segment extending $s$ does not lie on $s$, the closest point to $p$ on $s$ is an endpoint of $s$ and we have a pair of endpoints of distance at most $\frac{1}{2}\sqrt{3} \leq 1$.
    Otherwise, let $p'$ be the orthogonal projection of $p$ on $s$. Since the length of $s$ is at most $1$, there is an endpoint $q$ on $s$ s.t. $\|p'-q\|\leq 1/2$. As $\triangle p'pq$ is a right-angled triangle, Pythagoras implies $\|p-q\| = \sqrt{\|p'-q\|^2 + \|p'-p\|^2} \leq \sqrt{(1/2)^2 + (\frac{1}{2}\sqrt{3})^2} = 1$.
\end{proof}

Next we show that all candidate representatives are close to each other: a property which allows us to easily choose a representative locally in few rounds.

\begin{lemma}\label{lem:rep-candidates-3-hops}
	Let $u,v\in \mathcal C(g)$ be candidates for representing a grid node $g$. Then $\hop_G(u,v)\leq 3$.
\end{lemma}
\begin{proof}
	Firstly, we show that $\forall u, v \in \mathcal{C}_1(g), \hop_G(u, v) \leq 3$. Assume (w.l.o.g) that $u$ is a vertex of triangle $T_1$ and that $v$ is a vertex of triangle $T_2$, and that both $T_1$ and $T_2$ contain $g$.
	If the edges of these triangles do not intersect each-other, one triangle is contained in the other, so all vertices of the triangles are adjacent. Otherwise, at least two sides of the triangles intersect, so by Lemma~\ref{lem:close-segments-connect}, there are exist two vertices of $T_1$ and $T_2$ that are joined by an edge.
	Therefore, the induced subgraph on this pair of triangles has diameter at most $3$, so $\hop_G(u,v)\leq 3$.
	
	Next, we show that $\forall u, v \in \mathcal{C}_2(g), \hop_G(u, v) \leq 3$. All nodes in $\mathcal{C}_2(g)$ are incident to an edge of $G$ that intersects the cell $C$ which $g$ is in. Let $e_1,e_2$ be two such edges. Since they both intersect $C$, their distance is at most $\sqrt{2}\cdot c = \frac{1}{10}\sqrt{30} \leq \frac{1}{2}\sqrt{3}$. By Lemma~\ref{lem:close-segments-connect},
	$e_1$ and $e_2$ have two endpoints that are adjacent in $G$. Hence the subgraph induced by the endpoints of $e_1,e_2$ has diameter at most $3$ and $\hop_G(u,v)\leq 3$.
	
	Finally, we show that $\forall u \in \mathcal{C}_1(g), v \in \mathcal{C}_2(g), \hop_G(u, v) \leq 3$.
    Let $e$ be an edge incident to $v$ which intersects $C$, and let $T$ be a triangle containing $g$, of which $u$ is a vertex. There are two cases:
    \begin{itemize}
        \item Case $C \subseteq T$. In this case, $e$ also intersects $T$ in the interior or the boundary of $T$. If $e$ lies entirely within $T$ then both endpoints are adjacent to all vertices of $T$; if $e$ intersects an edge of $T$ then Lemma~\ref{lem:close-segments-connect} applied to $e$ and an edge of $T$ which $e$ intersects shows that an endpoint of $e$ is adjacent to a vertex of $T$.
        \item An edge $e_T$ of $T$ intersects $C$ (includes the case $T \subseteq C$). Since $e$ intersects $C$ as well we have that $\dist(e,e_T) \leq \sqrt 2 c \leq \frac{1}{2}\sqrt{3}$. Again, we apply Lemma~\ref{lem:close-segments-connect} on $e$ and $e_T$ to show that an endpoint of $e$ is adjacent to a vertex of $T$.
    \end{itemize}
    Since in both cases, a vertex of $e$ is adjacent to a vertex of $T$, we have $hop_G(u,v) \leq 3$.\qedhere
\end{proof}

\section{Additional Technical Details for Section~\ref{sec:gridrep}}
\label{app:additional_gridrep}

This section contains the technical details for the results in Section~\ref{sec:gridrep}. We show a variety of properties here which all relate to either the ease with which we can compute the representation $R$; or the simulation of it.
The first claim to prove is that each node in the UDG only has to represent grid nodes within a constant distance.

\begin{lemma}
\label{lem:candidates_are_close_to_grid_nodes}
    For all grid nodes $g \in V_\Gamma$ and for all nodes $v \in \mathcal{C}(g), \Vert v - g \Vert \leq (1 + \frac{c}{\sqrt{2}})$.
\end{lemma}
\begin{proof}
    If $v \in \mathcal{C}_1(g)$, then $\Vert v - g \Vert \leq 1$, because $g$ is inside a triangle with sides at most $1$ of which $v$ is a vertex.
    If $v \in \mathcal{C}_2(g)$, then $\Vert v - g \Vert \leq (1 + \frac{c}{\sqrt{2}})$. Let $C$ be the grid cell containing $g$. There must be an edge $\{u, v\} \in E$ which intersects $C$. The furthest that $v$ can be therefore is the largest distance away from $g$ in $C$ ($\frac{c}{\sqrt{2}}$) plus the maximum length of the edge ($1$).
\end{proof}

Because the grid granularity is constant, this implies the following observation:

\begin{corollary}
\label{cor:constant_responsibility}
    All nodes in the UDG represent at most a constant number of grid nodes.
\end{corollary}

Next we aim to compute the representatives of grid nodes efficiently. Let $\mathcal C(g)$ be the candidate nodes from Definition \ref{def:representatives}, i.e., all nodes that come into question as representative for some grid node $g \in V_\Gamma$. Then $\mathcal C(g)$ induces a constant diameter sub-graph by Lemma \ref{lem:rep-candidates-3-hops}. Hence, all nodes in $\mathcal C(g)$ can quickly agree who will represent $g$ using broadcast-based aggregation to constant depth.

For this task we define an algorithm $\mathcal A_g$ that computes the representative of $g \in V_\Gamma$. Further below, we will show that we can run all algorithms $\mathcal A_g, g \in V_\Gamma$ in parallel without causing too much congestion. Let us look at $\mathcal A_g$ in more detail.

\section{Computing and Simulating the Grid Graph}

This section focuses on computing a grid graph given a local network that forms a unit-disk graph with a given embedding.

\subsection{Algorithm $\mathcal A_g$ for Computing Representatives of $V_\Gamma$}

For 3 rounds we do the following: In the first round, all nodes $u \in V$ which are a member of the set of candidates $\mathcal C(g)$ for some grid node $g$ prepare a message $M_g(u)$ containing the coordinates of $g$, the coordinates of $u$, $ID(u)$, and a \emph{priority bit} $b(u)$, which is set to $1$ if $u \in \mathcal{C}_1(g)$, else $0$.
Then they send their respective message $M_g(u)$ to all their neighbours.
In the two remaining rounds, each node $u$ broadcasts to all its neighbours the set of messages $M_g(v)$, where $g$ is a grid node, and $v$ is the currently known ``best'' node to represent $g$ which $u$ has heard about so far. Note that the ``best'' node $v$ is defined as the node with the highest priority bit $b(u)$; ties are broken first by proximity to $g$ (closest wins), and then by ID.

\begin{lemma}
\label{lem:parallelize-algorithms}
    Let $\mathcal A_1, \ldots , \mathcal A_k$ be distributed \CONGEST algorithms in a UDG $G$ with the following property. There is a set of discs $D_1, \dots , D_k \subseteq \mathbb R^2$ such that in any point of $\mathbb R^2$ only a \emph{constant} number of the $D_i$ overlap and each $\mathcal A_i$ is restricted to $D_i$, i.e., it causes only nodes within $D_i$ to send or receive messages. Then we can run $\mathcal A_1, \ldots , \mathcal A_k$ in parallel in \CONGEST.
\end{lemma}

\begin{proof}
    The lemma formalizes the intuitive observation that if only a constant number of algorithms interfere with each other in a given region of $G$, then we can run these in parallel by combining messages that are send concurrently over an edge into a single one. The message size increases only by a constant factor. Note that one can also utilize time multiplexing to dilate the run-time by a constant factor instead of the message size.
\end{proof}

\begin{lemma}
\label{lem:decide_representatives}
    Let $g \in V_\Gamma$. We can compute the  %unique\ps{really unique?}
    representative $r$ of $g$ characterized in Definition \ref{def:representatives} in $\bigO(1)$ rounds.
\end{lemma}

\begin{proof}
    Firstly, we claim that after $3$ rounds the broadcast of algorithm $\mathcal A_g$ will send the correct representative of $g$ to all nodes in $\mathcal{C}(g)$. This follows from Lemma~\ref{lem:rep-candidates-3-hops} which states that the diameter of the subgraph induced on $\mathcal{C}(g)$ is at most $3$. As algorithm $\mathcal A_g$ ranks nodes in the same manner as Definition~\ref{def:representatives} the correct representative will propagate throughout $\mathcal{C}(g)$ in this number of rounds.

    We argue that Lemma \ref{lem:parallelize-algorithms} applies for the algorithm instances $\mathcal A_g, g \in V_\Gamma$. All nodes in $\mathcal{C}(g)$ are within distance $(1 + \frac{c}{\sqrt{2}})$ of $g$ (Lemma~\ref{lem:candidates_are_close_to_grid_nodes}). As defined above, the broadcast of algorithm $\mathcal A_g$ runs in $3$ rounds. Therefore a node can only participate in the broadcast of $\mathcal A_g$ if it is within distance $(4 + \frac{c}{\sqrt{2}})$ of $g$. Lemma \ref{lem:parallelize-algorithms} applies due to the minimum distance of $c$ between grid nodes. %clearly there are only constantly many grid nodes which satisfy this for any given node. Therefore in the worst case nodes must send a constant number of messages concerning a constant number of grid nodes over each edge, and this can be done in one round in \CONGEST.
\end{proof}

\subsection{Algorithm $\mathcal B_r$ for Computing Representations of $E_\Gamma$}

To complete the computation of a representation $R$ of $\Gamma$ we will show how to compute the representations of edges $\{g_1,g_2\} \in E_\Gamma$ which correspond to 3-hop paths $\Pi_{r_1,r_2}$ in $G$. Again, we start by describing an algorithm for that task.

Let $r \in V$ be the representative of some $g \in V_\Gamma$. In the first phase $\mathcal B_r$ does the following. It first runs a distributed breadth-first search with root node $r$ for 3 rounds in the local network to construct BFS-tree $T_r$ with a depth of 3 hops. Then each node $v \in V$ with $\hop_G(v,r) \leq 3$ learns ID($r$) and its parent $p_r(v)$ in the tree $T_r$. Assume that we run $\mathcal B_{r'}$ for the representatives $r'$ of \textit{all} grid nodes to construct all BFS-trees $T_{r'}$ in parallel.

The goal of the second phase of $\mathcal B_r$ is to add a path to all neighbor representatives $r'$ with higher ID.
Note that, in particular, $r$ learns ID($r'$) and $p_{r'}(r)$ of all representatives $r'$ of the grid nodes $g'$ that are adjacent to $g$ in $\Gamma$ (due to Lemma \ref{lem:rep_grid_3_hops}). If ID($r$) $<$ ID($r'$), then $r$ adds the edge $\{r, p_{r'}(r)\}$ to $E_R$. Then $r$ sends a message towards $p_{r'}(r)$ that instructs all nodes $v \neq r'$ in the branch of $T_{r'}$ from $r$ to $r'$ (including itself) to join $V_R$ and add the edge $\{v, p_{r'}(v)\}$ to $E_R$. After at most 3 rounds this message reaches $r'$ and phase 2 terminates.

\medskip

\begin{lemma}
    \label{lem:forward_message_in_R}
    Let $\{g,g'\} \in E_\Gamma$ and let $r,r' \in V_R$ be their representatives. One message from $r$ to $r'$ can be delivered in at most $\bigO(1)$ rounds for all such pairs $r,r'$ in parallel.
\end{lemma}

\begin{proof}
    This lemma encapsulates the intuitive observation that we can use the representation $\Pi_{r,r'}$ of $\{g,g'\} \in E_\Gamma$ for communication between $r$ and $r'$. There are a few details that need mentioning. First, from running the algorithms $\mathcal B_{r}, \mathcal B_{r'}$ the node $r$ and all nodes on $\Pi_{r,r'}$ are aware in which direction a message from $r$ to $r'$ has to be forwarded in order to reach $r'$. Second, although a node in $V_R$ might have to represent multiple grid nodes, there are only a constant number of those due to Corollary~\ref{cor:constant_responsibility}. So if $r$ has to send multiple messages in its role of being representative to multiple grid nodes, by Lemma \ref{lem:parallelize-algorithms} this can be done in parallel in $\bigO(1)$ rounds.
\end{proof}

\subsection{Extension to the \BCONGEST model}
\label{sec:works-in-broadcast}

We conclude this section by showing that the algorithms for constructing and simulating the grid graph that were presented above work even if the local mode of communication is in the more restrictive \BCONGEST model. In the \BCONGEST model, each node must send the \emph{same} message to all of its neighbours in any given round, rather than sending each neighbour a distinct message (as is allowed in \CONGEST). This message, like the messages in the \CONGEST model, may only be $O(\log n)$ bits in size. Note that if we can construct and simulate the grid graph \BCONGEST, allows us to also simulate any \CONGEST algorithm on the grid graph under the same restriction, as it has constant degree (of 4) anyway.

%\sco{Remark or lemma?}
\begin{lemma}
\label{rem:works-in-broadcast}
Lemma~\ref{lem:compute_representation} and Theorem~\ref{thm:hybrid_on_grid} hold even if the local communication mode of the \HYBRID model is modified to use the more restrictive \BCONGEST model.
\end{lemma}
\begin{proof}
    Firstly, note that, in the \BCONGEST model, if a node would like to send constantly many distinct messages to neighbours, it can achieve this via time multiplexing and suffer only a constant factor slowdown (or with a constant-factor increase in message size). Also recall that each node is the representative for at most a constant number of grid nodes (Corollary~\ref{cor:constant_responsibility}).

    The algorithms $\mathcal A_g$ to compute representatives of cells consist only of a series of broadcasts of the current ``best'' representative encountered so far and therefore Lemma~\ref{lem:decide_representatives} still holds in \BCONGEST.
    Algorithm $\mathcal B_r$ (which establishes edges to representatives of cells neighboring the cell of $r$) can also be run in the \BCONGEST model. The required observation is that the construction of a single constant-depth BFS tree requires only local broadcasts. Recall that, since every node is representative for only a constant number of grid cells, each node is the root of at most a constant number of these broadcasts. Using Lemma~\ref{lem:parallelize-algorithms}, each node will only have to participate in a constant number of these broadcasts. %Finally, when computing $E_R$, each node will only participate in a constant number of these operations and so this is also easy to implement.

    Lemma~\ref{lem:forward_message_in_R} remains true in \BCONGEST, since each node will be part of the underlying representation of at most a constant number of edges in $E_\Gamma$, thus each node can participate in delivering messages between adjacent grid nodes with only a constant factor slowdown. Theorem~\ref{thm:hybrid_on_grid} follows immediately.
\end{proof}

\section{Additional Technical Details for Section~\ref{sec:portals_and_labeling}}
 \label{app:techniques}
We give an overview on some techniques for hybrid networks that are used by our tree labelling algorithm in Section~\ref{sec:portals_and_labeling}.
A more detailed description of these techniques can be found in~\cite{FHS20,GHSW20,GmyrHSS17}.

\subsection{Pointer Jumping}\label{app:sec:pointer_jumping}
We show how to construct a network with diameter $\bigO(\log n)$ in time $\bigO(\log n)$ out of a simple line graph $L$ with $\bigO(n)$ nodes.
Assume each node $v \in L$ knows its left and right neighbor in the line (except for the left- and rightmost node, who only know one neighbor).
We let the nodes of $L$ generate \emph{shortcut edges} via \emph{pointer jumping}: In the first round, each node $v_i \in L$ that has two neighbors $v_{i-1},v_{i+1} \in L$ establishes the edge $\{v_{i-1},v_{i+1}\}$.
Whenever in each subsequent round, a node $v \in L$ receives two new shortcut edges $\{u,v\}, \{v,w\}$ in the previous round, $v$ generates another shortcut edge $\{u,w\}$.
It is easy to see that after $\bigO(\log n)$ rounds, no further shortcut edges are created and the resulting structure has diameter $\bigO(\log n)$, thus implying the following lemma.

\begin{lemma} \label{app:pointer_jumping}
	Given a line $L$ of $\bigO(n)$ nodes, setting up additional edges to obtain a structure $L^+$ with diameter $\bigO(\log n)$ and degree $\bigO(\log n)$ takes $\bigO(\log n)$ rounds.
\end{lemma}

Performing pointer jumping on each of our portals in the portal tree, the grid nodes within each portal $\mathcal P$ are able to set up a structure on which they can quickly broadcast information to all grid nodes within $\mathcal P$.
By doing so, we immediately obtain the following lemma.

\begin{lemma} \label{lemma:portal_broadcasting}
	Any $\bigO(\log n)$-bit message can be broadcast among all grid nodes within a single portal $\mathcal P$ in $\bigO(\log |\mathcal P|)$ rounds.
\end{lemma}

\subsection{Rooting Trees of Arbitrary Depth} \label{app:techniques:rooting}
Given a tree $T$ of $n$ nodes with arbitrary depth and constant node degree, we show how to root $T$ at the node $s$ with minimum identifier, such that every node in $T$ is aware of its parent node.
To do so, we adapt the well-known \emph{Euler tour} technique to a distributed setting.
Every node $v \in T$ with neighbors $v(0),\ldots,v(\deg(v)-1)$ (sorted in ascending order by their identifiers) simulates a virtual node $v_i$ for each of its neighbor $v(i)$.
We now connect all virtual nodes to a simple cycle $C$ as follows.
For every node  $v_i \in C$, there is an edge $(v_i,u_j) \in C$ such that $u = v((i+1) \mod \deg(v))$ and $v = u(j)$.
Therefore, each virtual node $v_i$ that belongs to the node $v$ with identifier $id(v)$ is able to introduce itself to its predecessor in $C$ by sending its \emph{virtual identifier} $\widetilde{id}(v_i) := id(v) \circ i$ for all $i \in [\deg(v)]$, where $\circ$ denotes the concatenation of two binary strings and $[k] = \{0,\ldots,k-1\}$.
Since each node simulates only a constant number of virtual nodes, the number of virtual nodes in the cycle $C$ is $\bigO(n)$.

We first describe how to determine the virtual node $s_i$ with minimal virtual identifier in $\bigO(\log n)$ rounds.\footnote{The virtual node with minimal virtual identifier $\widetilde{id}(s_i) = id(s) \circ i$ is the node $s_i$ with $id(s) \leq id(v)$ for all $v \in T$ and $i = 0$.}
Note that the node $s$ simulating $s_0$ is then the node with minimal identifier.
Consider the cycle $C$ of virtual nodes and denote the edges of the cycle as \emph{level-$1$} edges.
Our algorithm works in multiple iterations.
Initially, each virtual node $v$ stores its own virtual identifier $\widetilde{id}(v)$ in some variable $v.I$.
In the first iteration, each virtual node $v$ does the following.
In the first step, $v$ sends $v.I$ to its left neighbor in the cycle\footnote{The nodes may have different perceptions on which direction is left, but this is of no concern for our algorithm.}.
Upon receipt of a virtual identifier $v.I$, each node $u$ updates its variable $u.I$ to $v.I$ in case that $v.I < u.I$.
In the next step, each virtual node $v$ introduces its left neighbor $v_l$ to its right neighbor $v_r$ to create the edge $\{v_l,v_r\}$, a level-$2$ edge.
In each subsequent iteration, say the $i$-th iteration, each node $v$ first sends $v.I$ along with its own identifier via \emph{all} of its level-$j$ edges (for all $j \in \{1,\ldots,i\}$) and then creates level-$(i+1)$ edges, using its level-$i$ edges created in the previous iteration.
Note that after the $i$-th iteration, $2^i$ nodes are aware of $v$'s virtual identifier and thus have stored $v$'s virtual identifier in their variables $u.I$, in case $v$'s virtual identifier is the minimal virtual identifier among all of these nodes.
We proceed in this manner until a virtual node $v$ has received its own virtual identifier from its right neighbor in some iteration, as in this case all nodes have received $v$'s virtual identifier.
This happens at the node $s_0$ with minimum virtual identifier after $\bigO(\log n)$ rounds, because $C$ contains $\bigO(n)$ virtual nodes.
Thus, all that is left to do is to let $s_0$ announce itself as the root of the tree by broadcasting a message on the cycle with the generated shortcuts, indicating the termination of the algorithm and announcing itself as the node with minimum virtual identifier.
This takes another $\bigO(\log n)$ rounds.
Then, each virtual node is now aware of the node $s_0$ with minimum virtual identifier and therefore also of the node $s$ with minimum identifier.

Now we want to root the tree $T$ at $s$.
The virtual node $s_0$ starts broadcasting its virtual identifier via all of its outgoing edges to the left (including all of the generated shortcuts from before).
During this broadcast, we keep track of the traversal distance of the message to be broadcasted, such that each virtual node is able to determine how many hops it is away from $s_0$ in the cycle.
A real node $v$ can now determine its parent in the tree $T$ by looking at its virtual node with minimum traversal distance to $s_0$.
Let this node be the node $v_i$ and let $u_i$ be the predecessor of $v_i$ in the cycle $C$.
Then it is easy to see that $u$ is the parent node of $v$ in $T$, resulting in $T$ getting rooted at $s$ and implying the following lemma.

\begin{lemma} \label{app:tree_rooting}
    Let $T$ be a tree of $n$ nodes with constant node degree.
    $T$ can be rooted at the node $s$ with minimal identifier within $\bigO(\log n)$ rounds.
\end{lemma}

\subsection{Depth-First Search on Trees} \label{app:techniques:dfs}
Given a rooted tree $T$ of $n$ nodes with arbitrary depth and constant node degree, we compute for each node $v \in T$ the preorder number $l_v \in \mathbb{N}$ according to a depth-first search (DFS) of $T$.
Let $s$ be the root of $T$.
As the first step, we perform the distributed Euler-Tour technique described in the previous section with the exception that the virtual node $s_0$ refrains from introducing itself to its predecessor.
It is easy to see that this results in the virtual nodes being arranged in a simple line $L$ instead of a cycle.

Next, we apply the pointer jumping technique from Lemma~\ref{app:pointer_jumping} to transform $L$ into a structure $L^+$ with diameter $\bigO(\log n)$.
Through a single broadcast from the leftmost node $s \in L^+$, we are now able to compute a number $l'_u$ for a virtual node $u$ indicating the number of (real) nodes $v$ for which at least one of $v$'s virtual nodes is left of $u$ on the line $L$.
Each real node $u$ then sets $l_u$ to the minimum value out of all $l'_u$ values of its virtual nodes, which corresponds to $u$'s position in the DFS.

\begin{lemma} \label{app:dfs}
    Let $T$ be a rooted tree of $n$ nodes with constant node degree.
    A DFS on $T$ where each node $v \in T$ is assigned its number $l_v \in \mathbb{N}$ in the DFS can be computed in $\bigO(\log n)$ rounds.
\end{lemma}

\subsection{Computing the Maximum Preorder Number in a Rooted Tree} \label{app:techniques:max_preorder}
Assume we are given a rooted tree $T$ of $n$ nodes with arbitrary depth and constant node degree in which every node $v \in T$ possesses a preorder number $l_v \in \mathbb{N}$ according to a DFS.
We compute for each node $v \in T$ the maximum preorder number possessed by a node in $v$'s subtree, i.e., we compute $r_v = \max\{l_u \in \mathbb{N} \mid u \in T(v)\}$, where $T(v)$ is the subtree of $T$ with $v$ as the root.\footnote{A more general approach to this problem is presented in~\cite[Lemma 4.12]{GHSW20}, where the goal is to compute the value of a distributive aggregate function for each node $v$'s own subtree. An aggregate function $f$ is called \emph{distributive} if there is an aggregate function $g$ such that for any multiset $S$ and any partition $S_1,\ldots,S_\ell$ of $S$, $f(S)=g(f(S_1),\ldots,f(S_\ell))$. Classical examples are MAX, MIN, and SUM. However, due to the generality of $f$, the authors had to make use of randomization, which results in a runtime of $\bigO(\log n)$, w.h.p. for their algorithm. We present a deterministic $\bigO(\log n)$-algorithm that is specifically tailored to the MAX function in this section.}

Before we describe our algorithm, we let the nodes compute an upper bound of $\log n$, i.e., some value $d = \bigO(\log n), d \geq \log n$ as follows.
We compute the line $L$ via the Euler-tour described earlier on the rooted tree $T$ and apply Lemma~\ref{app:pointer_jumping} on $L$ to obtain the structure $L^+$.
Then we perform a broadcast from the rightmost node $u$ in $L^+$ to the leftmost node $s_0$ in $L^+$, where each message generated by the broadcast contains a counter that is incremented by $1$ once the message is forwarded.
The node $s_0$ then maintains a variable $d$ that contains the maximum counter received by $s_0$.
Since $L^+$ has diameter $\bigO(\log n)$, the broadcast finishes after $\bigO(\log n)$ rounds.
Once the broadcast is finished, it is easy to see that $d = \bigO(\log n)$.
The node $s_0$ then broadcasts $d$ to all nodes in $L^+$, such that after another $\bigO(\log n)$ rounds, each node knows $d$.
Observe that $d \geq \log D(T)$, where $D(T)$ is the depth of the tree $T$.

We are now ready to describe the algorithm for computing the values $r_v$ for each node $v \in T$.
Initially, each node $v$ sets $r_v$ to  $l_v$.
Denote the edges of $T$ as \emph{level-$0$} edges.
The algorithms performs $i = 1,\ldots,d$ iterations, each iteration needing $\bigO(1)$ rounds.
Iteration $i$ works as follows at each node $v \in T$.
First, if $v$ has a level-$(i-1)$ edge going up in the tree to some node $u$, then $v$ sends $r_v$ to $u$.
Upon receipt of a value $r_w$ from node $w$ in the previous step, $v$ updates $r_v$ by setting $r_v \gets r_w$ and \emph{marks} the edge $\{v,w\}$.\footnote{Note that in the first iteration ($i=1$), a node $v$ receives a value $r_w$ from each of its child nodes $w$. It then just sets $r_v$ to be the maximum value out of all received values $r_w$. It is easy to see that $v$ receives at most one message in any subsequent iterations in this step.}
As the final step of the iteration, $v$ checks whether it has a marked edge $\{v,u\}$ going up the tree and a marked edge $\{v,w\}$ going down the tree.
If that is the case, $v$ creates a level-$i$ edge $\{u,w\}$ by introducing $u$ to $w$ and vice versa.
If not, then $v$ marks itself as \emph{ready}.

Let $T(v)$ be the subtree of $T$ with $v$ as the root and let $w \in T(v)$ be the leaf node with maximum preorder number.
Consider the unique path $P$ up the tree from $w$ to $v$ in $T(v)$.
It is easy to see that our algorithm transfers the preorder number $l_w$ to all nodes on this path within $\lceil \log k \rceil$ iterations, where $k$ is the length of $P$, because in each iteration $i$, new level-$i$ shortcuts are added to the nodes on the path in a manner similar to the pointer-jumping approach from Section~\ref{app:sec:pointer_jumping}.
Therefore, once a node $v$ has marked itself as ready in iteration $i$, $v$ has received the desired value for $r_v$ in iteration $i$.
As each node $v \in T$ performs the algorithm in parallel, each node $v$ has determined $r_v$ after at most $d=\bigO(\log n)$ iterations (recall that $d \geq \log D(T)$).
Note that the node degree for each node $v$ does not exceed $\bigO(\log n)$ throughout the algorithm, as in each iteration, $v$'s degree increases by at most $2$.

We obtain the following lemma.

\begin{lemma} \label{lem:max_preorder_number}
    Let $T$ be a rooted tree of $n$ nodes with constant node degree in which every node $v \in T$ possesses a preorder number $l_v \in \mathbb{N}$ according to a DFS on $T$ starting at its root.
    Each node $v \in T$ can compute the value $r_v = \max\{l_u \in \mathbb{N} \mid u \in T(v)\}$, where $T(v)$ is the subtree of $T$ with $v$ as the root, within $\bigO(\log n)$ rounds.
\end{lemma}

\section{Additional Technical Details for Section~\ref{sec:grid_routing}}
\label{app:grid_routing}

\newcommand{\pushcode}[1][1]{\hskip\dimexpr#1\algorithmicindent\relax}
\begin{algorithm}[!ht]
\caption{Routing Strategy for a Grid Node $g \in V_{\Gamma}$}
\label{algo:routing}
\begin{algorithmic}[1]
% 	\State \textbf{Variables maintained by $g$:}
% 	\State \pushcode[0] $g.L \in \mathbb{N}^2$: $g$'s own interval given to it by the tree labeling
% 	\State \pushcode[0] $g.P \in \mathbb{N}^2$: The label of the root node of the portal containing $g$
% 	\State \pushcode[0] $g.N, g.S, g.E, g.W \in V_A \cup \perp$: $g$'s grid neighbors in north, south, east and west direction
% 	\Statex
% 	\State \textbf{Notation:}
% 	\State \pushcode[0] $g.P \times t.P := g.P \not \subseteq t.P \wedge t.P \not \subseteq g.P$
% 	\Statex
	\State \textbf{On receipt of a message $m$ with target $t$ at $g$:}
	\State \pushcode[0] \textbf{if}($g.L = t.L$)
	\State \pushcode[0] \ \ \ \ \Return ``Arrived at $t$''
	\State \pushcode[0] $f \gets $\textsf{TryHorizontal}($m$, $g$, $t$)
	\State \pushcode[0] \textbf{if}($f = false$)
	\State \pushcode[0] \ \ \ \ \textsf{GoVertical}($m$, $g$, $t$)
	\Statex
	\State \textsf{TryHorizontal}($m$, $g$, $t$)
	\State \pushcode[0] \textbf{if}($g.P = t.P$)
	\State \pushcode[0] \ \ \ \ \textbf{return} $false$ \Comment {Already at the correct portal}
	\State \pushcode[0] \textbf{if}($(g.P \subset g_W.P \subseteq t.P) \vee (t.P \subseteq g_W.P \subset g.P) \vee (g.P \incomp t.P \wedge g.P \subset g_W.P)$)
	\State \pushcode[0] \ \ \ \ Send $m$ to $g_W$
	\State \pushcode[0] \ \ \ \ \textbf{return $true$}
	\State \pushcode[0] \textbf{if}($(g.P \subset g_E.P \subseteq t.P) \vee (t.P \subseteq g_E.P \subset g.P) \vee (g.P \incomp t.P \wedge g.P \subset g_E.P)$)
	\State \pushcode[0] \ \ \ \ Send $m$ to $g_E$
	\State \pushcode[0] \ \ \ \ \textbf{return $true$}
	\State \pushcode[0] \textbf{return} $false$ \Comment{Did not manage to go west or east}
	\Statex
	\State \textsf{GoVertical}($m$, $g$, $t$)
	\State \pushcode[0] \textbf{if}($(g.L \subset g_N.L \subseteq t.L) \vee (t.L \subseteq g_N.L \subset g.L) \vee (g.L \incomp t.L \wedge g.L \subset g_N.L)$)
	\State \pushcode[0] \ \ \ \ Send $m$ to $g_N$
	\State \pushcode[0] \textbf{if}($(g.L \subset g_S.L \subseteq t.L) \vee (t.L \subseteq g_S.L \subset g.L) \vee (g.L \incomp t.L \wedge g.L \subset g_S.L)$)
	\State \pushcode[0] \ \ \ \ Send $m$ to $g_S$
\end{algorithmic}
\end{algorithm}

In this section, we show that our grid routing strategy is optimal and correct.
First, we show optimality and, by adding some simple arguments, conclude correctness afterwards.

To show optimality, we first make the following observations, which simply follow from the fact that the cell polygon $P'$ is simple and contains no holes (Lemma~\ref{lem:cell-polygon-simple}).
Therefore, the adjacency graph of the portals forms a tree, in which there is a unique simple path from $s$ to $t$, which is also a shortest path.

\begin{lemma} \label{obs:portal_sequence:opt}
    Let $s,t \in V_{\Gamma}$ and let $P_1(s,t), P_2(s,t)$ be any two \emph{optimal} paths from $s$ to $t$.
    Then the sequence of portals traversed by $P_1(s,t)$ and $P_2(s,t)$ is the same.
\end{lemma}

We say that a path $P(s,t)$ contains a portal $p$ multiple times, if the sequence of grid-nodes from $s$ to $t$ contains a node from some other portal $p'$ in between two nodes from $p$.

\begin{lemma} \label{obs:portal_sequence}
    Let $s,t \in V_{\Gamma}$ and let $P(s,t)$ be a path from $s$ to $t$.
    Assume the sequence of portals $\mathcal P(s,t)$ traversed by $P(s,t)$ does not contain a portal twice.
    Then $\mathcal P(s,t)$ is equal to the sequence of portals traversed by any optimal path from $s$ to $t$.
\end{lemma}

We can now prove that we always move optimally when going in horizontal direction.

\begin{lemma} \label{lem:routing:horizontal}
    If at a grid node $g$ the routing algorithm (Algorithm~\ref{algo:routing}) forwards a message $m$ horizontally (west or east direction), then there exists at least one optimal path from $g$ to $t$ that moves in the same direction.
\end{lemma}

\begin{proof}
    We prove the lemma by showing that the sequence of portals traversed via Algorithm~\ref{algo:routing} on the path $P(s,t)$ between two nodes $s,t \in V_{\Gamma}$ does not visit a portal twice, i.e., once $P(s,t)$ has left some portal $\mathcal P$, $\mathcal P$ is never visited again in $P(s,t)$.
    The lemma then follows from Lemmas~\ref{obs:portal_sequence:opt} and \ref{obs:portal_sequence}.

    Assume to the contrary that $P(s,t)$ visits grid nodes $p, v_1,\ldots,v_k, q$, where $p,q$ belong to the same portal $\mathcal P$ and all $v_i$ belong to different portals, i.e., $p.P = q.P$ and either $v_i.P \subset p.P$, $v_i.P \supset p.P$, or $v_i.P \incomp p.P$.
    W.l.o.g. assume that $p_W = v_1$. %\kh{also, in between a different portal must be visited (not all nodes can be of $P$)}
    We distinguish between the following three cases:
    \begin{itemize}
        \item[(i)] Assume that $p.P \subset t.P$.
        Since we route $m$ from $p$ to $p_W$, it follows that $p.P \subset p_W.P$.
        The only possibility to visit the portal $\mathcal P$ again is by going east at $v_k$, i.e., $v_{k_E} = q$.
        Note that since $p.P \subset t.P$, the routing never goes to a portal whose portal label is a superset of $t.P$.
        Therefore, at $v_k$ it holds that $v_k.P \subset t.P$ and, consequently, $v_k.P \subset v_{k_E}.P$.
        Putting everything together, we obtain $$p.P \subset v_1.P \subseteq \ldots \subseteq v_k.P \subset q.P \subseteq t.P.$$%\kh{technically, some of the labels could be the same, since you could stay in the same portal for some time}
        Since it holds that $p.P = q.P$, we arrive at a contradiction.
        \item[(ii)] The case where $p.P \supset t.P$ works analogously to case (i).
        \item[(iii)] Assume that $p.P \incomp t.P$, i.e., $p.P \not \subseteq t.P \wedge p.P \not \supseteq t.P$.
        By definition of Algorithm~\ref{algo:routing} $m$ only visits portals whose labels are a superset of $p.P$, until we arrive at a portal whose label is a superset of $t.P$.
        From this point on, $m$ never visits a portal whose portal label is incomparable to $t.P$, i.e., $m$ would never visit the grid node $q$.
        Note that this follows from the fact that the adjacency graph of portals is a tree, and the portal labels satisfy the tree labelling property of the portal tree, i.e., labels of children are subintervals of the labels of their parents and labels of parents are superintervals of their children's labels.
    \end{itemize}
\end{proof}

Combining Lemma~\ref{obs:portal_sequence} and Lemma~\ref{lem:routing:horizontal} implies the following corollary.

\begin{corollary}
\label{cor:same_portals}
    The sequence of portals traversed by a message $m$ via our routing strategy is the same as for any optimal path.
\end{corollary}

To also prove optimality in vertical direction, we need the following auxiliary lemma, which states the reverse direction of Lemma~\ref{lem:routing:horizontal}.

\begin{lemma} \label{lem:routing:vertical:aux}
	If at a grid node $g$ there is an optimal path that routes a message $m$ horizontally (west or east direction), then the routing algorithm (Algorithm~\ref{algo:routing}) routes $m$ horizontally in the same direction.
\end{lemma}

\begin{proof}
	Assume to the contrary that our routing strategy does not send $m$ horizontally at $g$, even though there exists an optimal path that does so.
	Since our routing strategy does not send $m$ horizontally, it sends $m$ vertically.
	By Corollary \ref{cor:same_portals}, it is guaranteed that our routing strategy will eventually switch to the same portal as the optimal path.
	However, due to the construction of our algorithm, the conditions for switching portals are exactly the same at each node in the portal where $g$ is contained (recall that we only compare the portal label with $t$'s portal label).
	Therefore, our routing strategy must have already sent $m$ horizontally at $g$.
\end{proof}

\begin{lemma} \label{lem:routing:vertical}
    If at a grid node $g$ the routing algorithm (Algorithm~\ref{algo:routing}) routes a message $m$ vertically (north or south direction), then there exists at least one optimal path from $g$ to $t$ that routes $m$ in the same direction.
\end{lemma}

\begin{proof}
	First, assume that $g$ is contained in the same portal as $t$ and w.l.o.g. assume that $g.L \subset t.L$.
	Then $g$ is contained in the subtree with root $t$.
	By definition of our algorithm, our routing strategy sends $m$ up in the tree at $g$, so $g$ forwards $m$ in the direction of $t$ in the portal, which is obviously optimal.
	
	Now assume that $g$ and $t$ are contained in different portals.
	As our routing strategy decided to move vertically, and, more specifically, to \emph{not} move horizontally, any optimal path also moves vertically at $g$, due to Lemma~\ref{lem:routing:vertical:aux}.
	All that is left to show is that our routing strategy routes $m$ vertically in the optimal direction.
	Assume to the contrary that the optimal path routes $m$ in opposite vertical direction compared to our routing path.
	It is easy to see that our routing path does not reverse a move, and therefore, both paths would change to different portals once they decide to route $m$ horizontally again -- contradicting Corollary \ref{cor:same_portals}.
\end{proof}

Lemma~\ref{lem:routing:horizontal} guarantees that we route a message optimally when we decide going in horizontal direction (east or west).
Lemma~\ref{lem:routing:vertical} guarantees that we route a message optimally when we decide going in vertical direction (north or south).
Therefore, we obtain the following lemma, implying that our routing scheme for the grid graph is optimal regarding efficiency.

\begin{lemma}[Optimality]\label{lem:grid_routing:optimality}
    If the routing algorithm (Algorithm~\ref{algo:routing}) forwards a message $m$ with target node $t$ from a grid node $g$ to a grid neighbor of $g$, then there exists at least one optimal path that forwards $m$ in the same direction.
\end{lemma}

All that is left to show for our routing strategy to be correct is to show that in case the message $m$ is at some grid node $g \neq t$, then the routing strategy always forwards $m$ to one of $g$'s grid neighbors and never remains at $g$.

\begin{lemma} \label{lem:grid_routing:move}
	Suppose a message $m$ with target $t$ is at some node $g \neq t$.
	Then the routing algorithm always forwards $m$ to one of $g$'s grid neighbors.
\end{lemma}

\begin{proof}
	Assume that $g.L \subset t.L$.
	Then, by definition of our labelling for the portal tree, $g$ has a vertical neighbor (either $g_N$ or $g_S$) $g'$ in the grid graph with $g.L \subset g'.L \subseteq t.L$.
	It follows that, in case $g$ does not forward $m$ to one of its horizontal neighbors, it instead forwards $m$ to $g'$.
	
	The cases $g.L \supset t.L$ and $g.L \incomp t.L$ work analogously to the case above.
\end{proof}

Combining Lemmas~\ref{lem:grid_routing:optimality} and \ref{lem:grid_routing:move} yields the following corollary.

\begin{corollary}[Correctness]
    For every source-destination pair $s,t \in V_\Gamma$, the routing algorithm (Algorithm~\ref{algo:routing}) determines a path in $G_\Gamma$ leading from $s$ to $t$.
\end{corollary}

\end{document}